\documentclass[11pt,a4paper]{article}%
\usepackage{amsfonts}
\usepackage{amsmath,mathrsfs}
\usepackage{amssymb}
\usepackage{graphicx}%
\usepackage{color}

\newtheorem{theorem}{Theorem}

\newtheorem{conjecture}[theorem]{Conjecture}
\newtheorem{corollary}[theorem]{Corollary}

\newtheorem{definition}[theorem]{Definition}
\newtheorem{example}[theorem]{Example}

\newtheorem{lemma}[theorem]{Lemma}

\newtheorem{proposition}[theorem]{Proposition}
\newtheorem{remark}[theorem]{Remark}

\newenvironment{proof}[1][Proof]{\noindent\textbf{#1.} }{\ \rule{0.5em}{0.5em}}

\definecolor{awesome}{rgb}{1.0, 0.13, 0.32}

\numberwithin{equation}{section}

\begin{document}

\pagenumbering{arabic}

\title{Short-time Fourier transform of the pointwise product of two functions with application to the nonlinear Schr\"odinger equation}
\author{Nuno Costa Dias{\thanks{ncdias@meo.pt}}
\and Jo\~{a}o Nuno Prata{\thanks{joao.prata@mail.telepac.pt }}
\and Nenad Teofanov{\thanks{nenad.teofanov@dmi.uns.ac.rs}}}

\maketitle

\begin{abstract}
We show that the short-time Fourier transform of the pointwise product of two functions $f$ and $h$ can be written as a suitable product of the short-time Fourier transforms of $f$ and $h$. The same result is then shown to be valid for the Wigner wave-packet transform. We study the main properties of the new products. Furthermore, we use these products to derive integro-differential equations on the time-frequency space equivalent to, and generalizing, the cubic nonlinear Schr\"odinger equation. We also obtain the Weyl-Wigner-Moyal equation satisfied by the Wigner-Ville function associated with the solution of the nonlinear Schr\"odinger equation. The new equation resembles the Boltzmann equation.
\end{abstract}

\section{Introduction}

Motivated by applications in mathematical physics we propose phase-space analogues of the
convolution identity for the Fourier transform
\begin{equation}
\widehat{f \cdot h}= \widehat{f} \star \widehat{h}~,
\label{eqIntro2}
\end{equation}
where  $f,h \in L^1 (\mathbb{R}^d)$, $\cdot$ is the pointwise product of functions, $\star$ is  the convolution,
and $ \;  \widehat{\cdot} \; $ denotes the Fourier transform.
In our study, the Fourier transform is replaced by the short-time Fourier transform (STFT) and the Wigner wave-packet transform,
while the convolution is replaced by the Gabor product $ \natural_g $ and the Wigner product $ \sharp_{g} $, respectively.

Given a fixed window $g \in S_0 (\mathbb{R}^d) \backslash \left\{0 \right\}$
(see Definition \ref{def:Feichtingeralgebra}), the STFT (see \eqref{eq2.1}) which maps
$S_0 (\mathbb{R}^d) $ into $ S_0(\mathbb{R}^{2d})$, $f \mapsto \left(V_g f\right)(x, \omega)$, gives some joint time-frequency description of a signal $f$. An inequality due to Lieb \cite{Lieb} permits the extension of the STFT to other function spaces. Let $r\geq 2$ and $\frac{1}{r}+\frac{1}{r^{\prime}}=1$.
If $f \in L^p(\mathbb{R}^d)$ with $r^{\prime} \leq p \leq r$, then there exists a constant $C>0$ (which depends on the window $g$), such that
\begin{equation}
\|V_g f \|_{L^r (\mathbb{R}^{2d})} \leq C \|f\|_{L^p(\mathbb{R}^d)}~.
\label{eqIntro4}
\end{equation}
In particular, for $p=r=2$, we have the equality
\begin{equation}
\|V_g f \|_{L^2 (\mathbb{R}^{2d})} = \|f\|_{L^2(\mathbb{R}^d)} \|g\|_{L^2(\mathbb{R}^d)}~,
\label{eq1AParseval}
\end{equation}
which is a consequence of the orthogonality relations (see Subsection \ref{subsec:transforms}).

By analogy with the Fourier transform, we introduce  the product $\natural_g: S_0 (\mathbb{R}^{2d}) \times S_0 (\mathbb{R}^{2d}) \to S_0 (\mathbb{R}^{2d})$, such that
\begin{equation}
V_g (f \cdot h)=\left(V_g f\right) \natural_g \left(V_g h\right)~,
\label{eqIntro5}
\end{equation}
and then try to extend it continuously to some function spaces $\mathcal{F}_0 (\mathbb{R}^{2d})$, $\mathcal{F}_1 (\mathbb{R}^{2d}) $, $\mathcal{F}_2 (\mathbb{R}^{2d})$, so that $\natural_g: \mathcal{F}_1 \times \mathcal{F}_2 \to \mathcal{F}_0$, with:
\begin{equation}
\|F_1 \natural_g F_2 \|_{\mathcal{F}_0} \leq C \|F_1\|_{\mathcal{F}_1} \|F_2\|_{\mathcal{F}_2}.
\label{eqIntro6}
\end{equation}
It is  clear that, in this analogy, (\ref{eqIntro4}) plays the r\^ole of the Hausdorff-Young inequality
$$
\|\widehat{f}\|_{L^{p^{\prime}}} \leq \|f\|_{L^p}~,
$$
when $p\in \left[1,2 \right]$,  $f \in L^p (\mathbb{R}^d)$, and $\frac{1}{p}+\frac{1}{p^{\prime}}=1$;
equation (\ref{eq1AParseval}) plays the r\^ole of Parseval's Theorem; equation (\ref{eqIntro5}) plays the r\^ole of the convolution identity (\ref{eqIntro2}); and  (\ref{eqIntro6}) plays the r\^ole of the Young convolution inequality (Theorem \ref{TheoremYoung}).

The product $\natural_g$ is introduced in Definition \ref{definitionGaborProduct} and the identity (\ref{eqIntro5}) is proved in Theorem \ref{theoremWindowedProduct1}. Various extensions of the form (\ref{eqIntro6}) are discussed in Subsection \ref{subsec:extension} (Theorems \ref{TheoremModulationSpace4}, \ref{thm:extension-version1} and Conjecture \ref{Conjecture1}). Modulation spaces, introduced by H.G. Feichtinger \cite{Feichtinger1,Feichtinger2} will play a crucial r\^ole in our derivations.

Similar considerations hold for the Wigner wave-packet transform and the corresponding  product $ \sharp_{g} $.

\par

Before giving a brief summary of the paper, let us explain the main motivation which comes from mathematical physics. We want to derive a phase space version of the nonlinear Schr\"odinger equation (NLSE). Equations of this form appear in many different situations. It may be a classical field equation with applications to optics and water waves (with its soliton solutions) \cite{Carles1}. NLSE's are used for modeling the propagation of deep-water wavetrains. Their doubly localized breather solutions can be connected to the sudden formation of extreme waves known as rogue or freak waves \cite{Vitanov}. But it may also be considered as  a nonlinear modification of the quantum mechanical Schr\"odinger equation. The nonlinearity is assumed to be of small magnitude in order not to violate the superposition principle in a dramatic way, and its purpose is, generally speaking, to induce the collapse of the wave function. A famous example is for instance the Schr\"odinger-Newton equation proposed by Roger Penrose \cite{Penrose} and Lajos Diosi \cite{Diosi}, where the nonlinearity is derived from Gauss's law for gravity. In this case, one would have a gravity induced collapse of the wave function.

The cubic NLSE is also used in Bose-Einstein condensate theory \cite{Kevrekidis}. Upon second quantization, one obtains a Bose gas of particles interacting through binary delta function interactions \cite{Korepin}. In one spatial dimension, this system is integrable and displays an infinite hierarchy of involutive integrals of motion.

In this work we will focus on the cubic NLSE. We shall try to rewrite equivalent equations in the corresponding phase-space. We believe that this may not be only of academic interest, but also of practical usefulness. Let us list some interesting characteristics of
phase-space representations: (i) Position and momentum (resp. time and frequency) appear on equal footing. (ii) In the case of the Weyl-Wigner formulation this leads to a beautiful symplectic/metaplectic covariance structure \cite{Folland,Gosson1,Gosson2}. (iii) Semiclassical (or high-frequency) limits are frequently addressed in this framework \cite{Carles2,Carles3,Carles4,Giulini}. (iv) Another reason for considering  phase-space representations comes from recent work that suggests that this seems to be a suitable framework for addressing hybrid quantum-classical systems appearing in molecular collision theory \cite{Bondar1,Bondar2}. (v) More specifically, the Weyl-Wigner formulation  for the NLSE leads to an equation which is strongly reminiscent of kinetic equations such as Boltzmann's equation. One is then in a position to apply analytical as well as numerical techniques developed in the context of kinetic theory to the NLSE. (vi) The Weyl-Wigner representation of the NLSE (and also the other phase-space representations) is equivalent to the standard NLSE if one plugs in as the initial datum the Wigner function associated with the initial wave function. However, the initial value problem in phase-space may be well-posed for more general initial distributions. This means that the phase-space formulation "contains" the solutions of the standard NLSE, but it is much richer, as it admits other solutions. This has been the {\it rationale} for a series of papers by the authors \cite{Dias1,Dias2,Dias3}, where the phase-space extensions are explored in several different contexts. We refer to \cite[Chapter 7]{Benyi} for an overview of results related to well-posedness of the NLSE in the framework of modulation spaces, see also the recent contributions \cite{OhWang, OhWang2}.

The paper is essentially divided in two parts. The first part (Section \ref{sec:2}) is self-contained and certainly interesting strictly from the point of view of harmonic analysis. Its aim is to derive the properties and extensions of products $\natural_g$  and $\sharp_{g}$ mentioned above. To that end we employ the powerful machinery of modulation spaces which are introduced in Subsection \ref{subs:modsp}. Apart from the well-known facts, we consider the range of the STFT and prove a representation and a density theorem (Theorems \ref{TheoremDensityWindows} and \ref{Theorem2}). We proceed with possible extensions of the products $\natural_g$ and $\sharp_{g}$ in Subsection \ref{subsec:extension}. By using different techniques we obtain partially overlapping results
(Theorems \ref{TheoremModulationSpace4} and \ref{thm:extension-version1}),  which we compare in Remark \ref{rem:extensions}. Finally, we determine an involution suitable for the products $\natural_g$ and $\sharp_{g}$ in Subsection \ref{subsec:algebras}.

In the second part of the paper (Section \ref{sec:3}), we derive three different representations of the cubic NLSE in phase-space. The first two are obtained via two windowed transforms: the STFT and the Wigner wave-packet transform. The main advantages of the STFT are: (i) it is the phase-space representation closest to the Fourier transform, and (ii) there is an enormous amount of analytical and numerical techniques readily available (see e.g. \cite{Grochenig}). The draw-back of the STFT as compared with the windowed wave-packet transform is the fact that the latter transform is symplectic/metaplectic covariant \cite{Gosson1,Gosson2,Wong}.
The third representation is the Wigner transform. The resulting equation is akin to the Boltzmann equation: it has a bilinear "collision" term. However, this term is nonlocal in the spatial variable (unlike the Boltzmann equation).

\subsection*{Notation}

We denote by $x \in \mathbb{R}^d$ a "time" (or position) variable and by $\omega \in \mathbb{R}^d$ a "frequency" (or momentum) variable, and  write $|x|= (x_1^2 + \cdots x_d^2)^{1/2}$. Functions on $\mathbb{R}_x^d$ are denoted $f,g,h, \dots$, and for those on $\mathbb{R}^{2d} \simeq \mathbb{R}_x^d \times \mathbb{R}_{\omega}^d$ we shall use capital letters $F,G,H, \dots$.
For a given function $f$ its complex conjugate, reflection, and involution are respectively given by
\begin{equation}
\overline{f} (x) = \overline{f (x) }, \;\;\;
\mathcal{I} f (x) = f(-x), \;\;\; \text{and} \;\;\;
f^\dagger (x) = \overline{f(-x)}.
\label{eqReflectionOperator}
\end{equation}
The Fourier transform of $f \in L^1 (\mathbb{R}^d)$ given by
\begin{equation}
\widehat{f} (\omega) := \int_{\mathbb{R}^d} f(x) e^{-2 i \pi x \cdot \omega} dx, \qquad \omega \in \mathbb{R}^d,
\label{eqFourierTransform}
\end{equation}
extends to $ L^2 (\mathbb{R}^d)$ by standard approximation procedure.
$ C_0 ^\infty (\mathbb{R}^d)$ is the space of smooth compactly supported functions,  $\mathcal S (\mathbb{R}^d)$ is the Schwartz space of smooth rapidly  decreasing functions and its dual $\mathcal S^{\prime} (\mathbb{R}^d)$ is the space of tempered distributions.
By $  {\mathcal S} ^{(1)} (\mathbb{R}^d ) $ we denote the Gelfand-Shilov space of smooth functions
given by:
\begin{equation}
f \in  {\mathcal S}^{(1)} (\mathbb{R}^d) \Leftrightarrow
\| f(x)  e^{h\cdot |x|}\|_{L^\infty} < \infty \;
\; \text{and} \;\;
\| \hat f (\omega)  e^{h\cdot |\omega|}\|_{L^\infty}< \infty, \;\; \forall  h > 0.
\label{eqGelfandShilovspace}
\end{equation}
Any $ f \in  {\mathcal S}^{(1)} (\mathbb{R}^d) $ can be extended
to the complex domain as holomorphic functions in a strip, \cite{GS}.
The dual space of  $  {\mathcal S} ^{(1)} (\mathbb{R}^d ) $ will be denoted by
$ ({\mathcal S}^{(1)})^{\prime} (\mathbb{R}^d ). $

We shall use the notation
$< \cdot, \cdot> $ for the duality bracket between a space of distributions $ \mathcal A^{\prime} (\mathbb{R}^d) $ and its test function space $ {\mathcal A} (\mathbb{R}^d)$,
and $\langle \cdot, \cdot \rangle_{L^2 (\mathbb{R}^d)}$ denotes the inner product in $L^2 (\mathbb{R}^d)$. Thus $< \cdot, \cdot> = \langle  \cdot, \overline{\cdot}  \rangle_{L^2 (\mathbb{R}^d)}$.
  The norm on the Lebesgue space $L^p (\mathbb{R}^d)$ is denoted by $|| \cdot ||_{L^p}$, $ 1\leq p \leq \infty$. Sometimes, if we need to emphasize the dimension, we write $|| \cdot ||_{L^p (\mathbb{R}^d)}$. Mixed-norm Lebesgue spaces $L^{p,q} (\mathbb{R}^{2d})$, $ 1\leq p,q < \infty$, consist of all $F \in \mathcal(\mathcal S^{(1)})^{\prime}  (\mathbb{R}^{2d})$ such that
\begin{equation}
||F||_{L^{p,q}} := \left( \int_{\mathbb{R}^d} \left(  \int_{\mathbb{R}^d} \left|F (x, \omega) \right|^p dx \right)^{\frac{q}{p}} d \omega\right)^{\frac{1}{q}}
\label{eqMixedNormSpace} < \infty.
\end{equation}
If $p= \infty $ or $q= \infty$, then the sup-norm is used.

If there exists $C > 0$ such that $a \le C b $ for some quantities $a$ and $b$, then we shall write $a \precsim b$. If $a \precsim b$ and $b \precsim a$, then we write $a\asymp b$.

The most important operators in time-frequency analysis are  the translation (or time-shift) operator
\begin{equation}
T_x f (y) = f(y -x),
\label{eqTranslationOperator}
\end{equation}
and the modulation (or frequency-shift) operator
\begin{equation}
M_{\omega} f(x) = e^{2 \pi i \omega \cdot x} f(x),
\label{eqModulationOperator}
\end{equation}
which act unitarily on $L^2 (\mathbb{R}^d)$, and satisfy the commutation relation
\begin{equation}
T_x M_{\omega} = e^{- 2 i \pi x \cdot \omega} M_{\omega} T_x.
\label{eqCommutationRelations}
\end{equation}
We also recall the (unitary) dilation operator:
\begin{equation}
D_s f(x) = s^{- d/2} f(s^{-1} x), \qquad  s >0.
\label{eqDilationOperator}
\end{equation}

\section{Time-frequency analysis} \label{sec:2}

In this section we start by recapitulating basic properties of the STFT, the Wigner wave-packet transform and the Wigner transform. Then we present two products $\natural_g$ and $\sharp_{g}$  in phase-space which generalize the convolution, and discuss their properties. To extend those products we employ Feichtinger's modulation spaces which we recall in Subsection \ref{subs:modsp}.
There we also consider the range of the STFT to some extent. Different extensions of the products are discussed in Subsection \ref{subsec:extension}, and their algebraic properties are treated in Subsection \ref{subsec:algebras}.

\par

The product and the convolution product are well defined under the conditions of the following theorem
(cf. \cite{Adams,Beckner,Brascamp}).

\begin{theorem}\label{TheoremYoung}
Let $ p, q, r \in [1,\infty].$
If $f \in L^p(\mathbb{R}^d)$ and $g \in L^{p^{\prime}}(\mathbb{R}^d)$ with $\frac{1}{p} + \frac{1}{p^{\prime}}=1$,
then $f \cdot g \in L^1 (\mathbb{R}^d)$ and the H\"older inequality holds:
\begin{equation*}
||f  g||_{L^1} \le  ||f||_{L^p} ||g||_{L^{p^{\prime}}},
\label{eqHolder}
\end{equation*}

If $f \in L^p(\mathbb{R}^d)$, $g \in L^q(\mathbb{R}^d)$ and $\frac{1}{p} + \frac{1}{q} = 1+ \frac{1}{r}$, then
$f \star g \in L^r(\mathbb{R}^d)$ and the Young inequality holds:
\begin{equation}
||f \star g||_{L^r}  \precsim ||f||_{L^p} ||g||_{L^q}.
\label{eqYoung}
\end{equation}
Here $\star $ denotes
the convolution product of $f$ and $g$:
\begin{equation}
\left( f \star g \right) (x)  = \int_{\mathbb{R}^d } f(x- y) g(y) dy
= \langle T_x \mathcal{I} f, \overline{g} \rangle_{L^2 (\mathbb{R}^d)}.
\label{eqConvolutionInnerProduct}
\end{equation}


\end{theorem}

The Fourier transform of the convolution is tantamount to point-wise multiplication.

\begin{theorem}\label{theoremFourierTransformConvolution}
Let $f, g \in L^1 (\mathbb{R}^d)$. Then
\begin{equation}
\left(f \star g \right)^{\widehat{}} = \widehat{f} \cdot \widehat{g}.
\label{eqFourierTransformConvolution1}
\end{equation}
If, additionally $\widehat{f}, \widehat{g} \in L^1 (\mathbb{R}^d)$, then
\begin{equation}
\left(f \cdot g \right)^{\widehat{}} = \widehat{f} \star \widehat{g}  .
\label{eqFourierTransformConvolution2}
\end{equation}
\end{theorem}

The following identities, which follow from the definitions above, will also be useful.

\begin{proposition}\label{propositionFourierTransform}
Let $f \in L^2 (\mathbb{R}^d)$. Then the following identities are valid
\begin{equation}
\left(\overline{f} \right)^{\widehat{}} = \overline{\mathcal{I} \widehat{f}} , \hspace{0.5 cm} \left(T_x f \right)^{\widehat{}} = M_{-x} \widehat{f}, \hspace{0.5 cm} \left(M_{\omega} f \right)^{\widehat{}} = T_{\omega} \widehat{f}, \hspace{0.5 cm} \left(D_s f \right)^{\widehat{}} = D_{\frac{1}{s}} \widehat{f},
\label{eqFourierTransformProperties}
\end{equation}
for every $ x,\omega \in \mathbb{R}^d,$ and $ s>0$.

\end{proposition}

\subsection{Windowed transforms} \label{subsec:transforms}

In this work we shall consider two time-frequency (phase-space) windowed representations of a function $f \in L^2 (\mathbb{R}^d)$: the Wigner wave-packet transform and the short-time Fourier transform.

Let $f,g \in L^2 (\mathbb{R}^d)$. The cross-Wigner transform is given by:
\begin{equation}
W(f,g) (x, \omega):= \int_{\mathbb{R}^d} f \left(x + \frac{y}{2} \right) \overline{g \left(x - \frac{y}{2} \right)} e^{- 2 \pi i \omega \cdot y} dy.
\label{eqCrossWignerFunction}
\end{equation}
If $g \in L^2 (\mathbb{R}^d)$ is a fixed window then the Wigner wave-packet transform is defined by \cite{Gosson1,Nazaikiinskii}:
\begin{equation}
\begin{array}{c}
f \in L^2 (\mathbb{R}^d) \mapsto W_g f (x, \omega) := 2^{- d} W (f,g) \left(\frac{x}{2}, \frac{\omega}{2} \right) = \\
\\
= 2^{-d} \int_{\mathbb{R}^d} f \left( \frac{x+y}{2} \right) \overline{g \left( \frac{x-y}{2} \right)} e^{-  \pi i \omega \cdot y} dy.
\end{array}
\label{eqWavepacketTransform}
\end{equation}
Likewise, the short-time Fourier transform (STFT) of $f$ with respect to the window $g$ is defined by
\begin{equation}
V_g f(x, \omega) := \int_{\mathbb{R}^d} f(t) \overline{g(t-x)} e^{-2 \pi i \omega \cdot t} dt
= \langle f, M_{\omega} T_x g \rangle_{L^2 (\mathbb{R}^d)}.
\label{eq2.1}
\end{equation}
If $g \in \mathcal S^{(1)} (\mathbb{R}^d) \backslash \left\{0 \right\}$, then
$V_g $ restricts to a continuous mapping from $ \mathcal S^{(1)} (\mathbb{R}^d) $ to $ \mathcal S^{(1)} (\mathbb{R}^{2d})$,
cf. \cite{To11}.

The Wigner wave-packet transform and the STFT  are related by:
\begin{equation}
W_g f (x, \omega) = e^{ i \pi  x \cdot \omega} V_{\mathcal{I} g} f (x,  \omega).
\label{eqRelationSTFTCrossWignerFunction}
\end{equation}

We note in passing that the STFT of $f$ and the Wigner wave-packet transform with respect to the window $g$ can also be expressed in terms of the Fourier transforms $\widehat{f}, \widehat{g}$ as
\begin{equation}
V_g f(x, \omega) = e^{-2 \pi i x \cdot \omega} V_{\widehat{g}} \widehat{f} (\omega, -x)
\label{eq2.2}
\end{equation}
and
\begin{equation*}
W_g f (x, \omega) = W_{\widehat{g}} \widehat{f} (\omega, -x).
\label{eq2.2.1}
\end{equation*}
Formula \eqref{eq2.2} is also known as {\em the Fundamental Identity of Time-Frequency Analysis}, \cite{Cordero1, Grochenig}.

One of the most remarkable facts about the STFT and the Wigner wave-packet transform are the following orthogonality relations (Parseval's identity) and Moyal's formula, respectively. Proofs can be found in e.g. \cite{Cordero1, Grochenig}.

\begin{theorem}\label{theorem2.1}
{\bf (Orthogonality relations for STFT).} Let $f_1,f_2,g_1,g_2 \in L^2 (\mathbb{R}^d)$; then $V_{g_j } f_j \in L^2 (\mathbb{R}^{2d})$ for $j=1,2$, and
\begin{equation}
  \langle V_{g_1}f_1 , V_{g_2}f_2  \rangle_{L^2 (\mathbb{R}^{2d})} =
   \langle f_1,f_2  \rangle_{L^2 (\mathbb{R}^{d})}   \langle g_2,g_1  \rangle_{L^2 (\mathbb{R}^{d})}.
\label{eq2.3}
\end{equation}
\end{theorem}

\begin{theorem}\label{theoremMoyalsFormula}
{\bf (Moyal's identity).} For  $f_1,f_2,g_1,g_2 \in L^2 (\mathbb{R}^d)$,
\begin{equation}
\begin{array}{c}
 \langle W_{g_1} f_1, W_{g_2} f_2  \rangle_{L^2 (\mathbb{R}^{2d})} =
  \langle W(f_1,g_1), W(f_2,g_2)  \rangle_{L^2 (\mathbb{R}^{2d})} =\\
 \\
 =  \langle f_1,f_2  \rangle_{L^2 (\mathbb{R}^{d})}  \langle g_2,g_1  \rangle_{L^2 (\mathbb{R}^{d})}.
\end{array}
\label{eqMoyalsIdentity}
\end{equation}
\end{theorem}

The STFT has the following behavior under time-frequency shifts, reflections and complex-conjugation.

\begin{proposition}\label{proposition2.2}
Let $f, g \in L^2 (\mathbb{R}^d)$. Then
\begin{equation}
V_g \left(T_u M_{\eta} f \right) (x, \omega) = e^{-2 i \pi u \cdot \omega} V_g f (x-u, \omega- \eta),
\label{eq2.4}
\end{equation}
\begin{equation*}
V_g \left(\mathcal{I} f \right) (x, \omega) = V_{\mathcal{I}g} f (-x, - \omega),
\label{eq2.5}
\end{equation*}
and
\begin{equation}
\overline{V_g f (x, \omega)} = V_{\overline{g}} \overline{f} (x, - \omega).
\label{eq2.5.1}
\end{equation}
\end{proposition}

Likewise, using \eqref{eqRelationSTFTCrossWignerFunction}, we have for the Wigner wave-packet transform:
\begin{proposition}\label{proposition2.2.1}
Let $f, g \in L^2 (\mathbb{R}^d)$. Then
\begin{equation*}
W_g \left(T_u M_{\eta} f \right) (x, \omega) = e^{ i \pi  (x \cdot \eta - u \cdot \eta - u \cdot \omega)}  W_g  f (x-u, \omega- \eta),
\label{eq2.4.1}
\end{equation*}
\begin{equation*}
W_g \left(\mathcal{I} f \right) (x, \omega) = W_{\mathcal{I} g} f (-x, - \omega),
\label{eq2.4.2}
\end{equation*}
and
\begin{equation}
\overline{W_g f (x, \omega)} = W_{\overline{g}}\overline{f}(x, - \omega).
\label{eq2.4.3}
\end{equation}
\end{proposition}

\begin{remark}
The cross-Wigner transform of $ f \in L^2 (\mathbb{R}^d)$ is defined in \cite{Gosson2, Gosson3} by means of the Grossmann-Royer operator
$$
\left(R  (x,\omega)f\right) (t) = e^{4\pi i\omega \cdot (t-x)} f(2x - t)
= e^{-4\pi i\omega \cdot  x}  \left(M_{2\omega} T_{2x}\mathcal{I} f\right) (t), \;\;\; x,\omega \in \mathbb{R}^d.
$$
Such operators originate from the problem of physical interpretation of the Wigner transform $W(f,f)$, \cite{Grossmann, Royer}.
The related Grossmann-Royer transform  $ R_g f (x,\omega) =  \langle R (x,\omega) f, g \rangle_{L^2 (\mathbb{R}^{d})} $ considered in \cite{Teofanov1}
is essentially the cross-Wigner transform: $W(f,g) (x,\omega)   =  2^d R_g f (x,\omega)$. We also mention that
$ W(f,g) (x,\omega)  =  \widehat{ A(f,g) } (x,\omega), $ where
$$
A (f,g) (x, \omega) = \int e^{-2\pi i \omega t}
f(t+ \frac{x}{2}) \overline{ g(t- \frac{x}{2})} dt, \;\;\; x,\omega \in \mathbb{R}^d.
$$
is the cross-ambiguity function of $f$ and $g$, see \cite{Gosson3}.
\end{remark}

\subsection{The Wigner transform}

Let $f \in L^2 (\mathbb{R}^d)$. The Wigner transform of $f$ is given by $Wf=W(f,f)$:
\begin{equation*}
f \mapsto Wf (x, \omega) := \int_{\mathbb{R}^d} f \left(x+ \frac{y}{2} \right) \overline{f \left(x- \frac{y}{2} \right)} e^{-2 i \pi \omega \cdot y} dy.
\label{eqWignerfunction1}
\end{equation*}
One calls $Wf$ the Wigner or Wigner-Ville function of $f$.

The Wigner function is akin to a joint probability density for position and momentum. Indeed, it is a real valued and normalized function (if $||f||_{L^2}=1$):
\begin{equation*}
\int_{\mathbb{R}^{2d}} Wf (x, \omega) dx d\omega = ||f||_{L^2}^2 =1,
\label{eqWignerfunction2}
\end{equation*}
and its marginal distributions are {\it bona fide} probability densities for position and momentum:
\begin{equation}
\begin{array}{l}
\int_{\mathbb{R}^d} Wf (x, \omega) d \omega = |f(x)|^2 \ge 0 , ~ \forall x \in \mathbb{R}^d\\
\\
\int_{\mathbb{R}^d} Wf (x, \omega) dx = |\widehat{f}(\omega)|^2 \ge 0, ~ \forall \omega \in \mathbb{R}^d,
\end{array}
\label{eqWignerfunction3}
\end{equation}
provided $f, \widehat{f} \in L^1 (\mathbb{R}^d) \cap L^2 (\mathbb{R}^d)$.

Moreover, one can compute all the probabilities according to the rules of quantum mechanics from the knowledge of the Wigner function. Finally, for any observable $A$ (a self-adjoint operator acting on the Hilbert space $L^2 (\mathbb{R}^d)$) with Weyl symbol $a (x, \omega)$, the expectation value is evaluated according to the following suggestive formula:
\begin{equation}
E(A) = \langle A f,f \rangle_{L^2 (\mathbb{R}^d)} = \int_{\mathbb{R}^{2d}} a(x, \omega) W f (x, \omega) d x d \omega.
\label{eqWignerfunction4}
\end{equation}
However, the Wigner function fails to be a full-fledged probability density, as it may take on negative values. These "negative probabilities" are a manifestation of quantum interference. Indeed, uncertainty principles preclude a simultaneous sharp localization of position and momentum, which is presupposed of a joint probability density. Non-negative Wigner functions are not excluded, but, as stated by Hudson's Theorem (see \cite{Hudson, Janssen, Toft2006}), they constitute a very restrictive class:

\begin{theorem}\label{HudsonTheorem}
{\bf (Hudson)} Let $f \in L^2 (\mathbb{R}^d)$. Then the corresponding Wigner function $Wf$ is everywhere non-negative if and only if $f$ is a generalized Gaussian.
\end{theorem}

It is also noteworthy that Wigner functions are square-integrable (cf.(\ref{eqMoyalsIdentity})):
\begin{equation*}
||Wf||_{L^2(\mathbb{R}^{2d})}^2 = \int_{\mathbb{R}^{2d}} |Wf (x, \omega)|^2 dx d \omega = ||f||_{L^2 (\mathbb{R}^d)}^4.
\label{eqWignerfunction5}
\end{equation*}
One calls $||Wf||_{L^2(\mathbb{R}^{2d})}^2$ the purity of the state $f$.

Before we proceed let us state the following useful identity which is easily obtained from \eqref{eqCrossWignerFunction} by the Fourier inversion formula:
\begin{equation}
f \left( x +  \frac{y}{2}\right) \overline{f \left( x -  \frac{y}{2}\right)} = \int_{\mathbb{R}^d} Wf (x, \omega) e^{2 i \pi \omega \cdot y} d \omega .
\label{eqWignerfunction6}
\end{equation}

\subsection{The products} \label{subs:products}

In this subsection we  introduce two products related to time-frequency representations.
We first introduce  the Feichtinger algebra $ S_0 (\mathbb{R}^{d})$ as the natural
framework for the definition of the Gabor product, and consider different extensions in Subsection \ref{subsec:extension}.

\begin{definition}\label{def:Feichtingeralgebra}
Let there be given $g \in \mathcal S^{(1)} (\mathbb{R}^{d}) \backslash \left\{0 \right\}$.
The Feichtinger algebra $ S_0 (\mathbb{R}^{d})$ consists of all $f\in L^2 (\mathbb{R}^{d}) $
such that
\begin{equation}
\| V_g f \|_{L^1 (\mathbb{R}^{2d})} < \infty.
\label{eq:Feichtingeralgebra}
\end{equation}
\end{definition}

It can be proved that  $ S_0 (\mathbb{R}^{d})$ is a Banach space with the norm
\begin{equation}
\| f \|_{S_0, g } = \| V_g f \|_{L^1 (\mathbb{R}^{2d})}.
\label{eq:Feichtingeralgebranorm}
\end{equation}
The Feichtinger algebra enjoys numerous properties useful for applications in time-frequency analysis, cf.
the recent survey \cite{Jakobsen} and references given there. Here we focus on the most important properties of $ S_0 (\mathbb{R}^{d})$
which will be  used in the sequel.

\begin{lemma} \label{lm:Szero}
{\bf (Basic properties of $ S_0 $ )}  Let  $ g \in  S_0 (\mathbb{R}^{d})\backslash \left\{0 \right\}$ be given.
Then the following is true:
\begin{itemize}
\item[i)] $ S_0 (\mathbb{R}^{d})$ is the smallest Banach space invariant under translations and modulations, and it is continuously embedded as a dense subspace in
$ L^2 (\mathbb{R}^{d}) $.
\item[ii)] If $f \in  S_0 (\mathbb{R}^{d})$, then $f$ is continuous, and $f, \hat f \in L^1 (\mathbb{R}^{d})$.
\item[iii)] If $f \in  S_0 (\mathbb{R}^{d})$, then its complex conjugation $\overline{f}$ , reflection $\mathcal{I} f$, and involution  $f^\dagger$ are also in  $S_0 (\mathbb{R}^{d})$ and
$$
\| f \|_{S_0, g } = \| \overline{f} \|_{S_0, g } = \| \mathcal{I} f \|_{S_0, g } = \| f^\dagger \|_{S_0, g } .
$$
\item[iv)] If $f \in  S_0 (\mathbb{R}^{d})$, then $ \hat f \in S_0 (\mathbb{R}^{d})$ and
$
\| f \|_{S_0, g } = \| \hat{f}\|_{S_0 , g }.
$
Thus, the Fourier transform is an isometry of $ S_0 (\mathbb{R}^{d})$.
\item[v)] For $f \in  L^2 (\mathbb{R}^{d})$, we have that $f \in  S_0 (\mathbb{R}^{d})$ if and only if
$ V_h f \in  S_0  (\mathbb{R}^{2d}) $ for some (and then all) $ h \in  S_0 (\mathbb{R}^{d})\backslash \left\{0 \right\}$,
and each such $h$ defines an equivalent norm on $ S_0 (\mathbb{R}^{d})$ via
$ \displaystyle  \| f \|_{S_0, h } = \| V_h f \|_{L^1 (\mathbb{R}^{2d})}. $
\item[vi)] $f \in  S_0 (\mathbb{R}^{d})$ if and only if
$  V_f f \in  S_0  (\mathbb{R}^{2d}).$
\item[vii)] $ S_0 (\mathbb{R}^{d})$ is closed under pointwise multiplication and convolution: if $f_1, f_2 \in  S_0 (\mathbb{R}^{d})$, then $f_1 \cdot f_2, f_1 \star f_2 \in  S_0 (\mathbb{R}^{d})$. Specifically, for $ g \in  S_0 (\mathbb{R}^{d})\backslash \left\{0 \right\}$,
\begin{equation}
\begin{array}{c}
\| f_1 \cdot f_2 \|_{S_0, g } \leq \| \hat g \|^{-1} _{\infty} \| f_1  \|_{S_0, g } \|  f_2 \|_{S_0, g },
\\
\| f_1 \star f_2 \|_{S_0, g } \leq \| g \|^{-1} _{\infty} \| f_1  \|_{S_0, g } \|  f_2 \|_{S_0, g }.
\end{array}
\label{eq:productconvolution}
\end{equation}
\item[viii)] The tensor product  $ \otimes :  S_0 (\mathbb{R}^{d}) \times  S_0 (\mathbb{R}^{d}) \rightarrow  S_0 (\mathbb{R}^{2d}) $    is a bounded bilinear operator.
\item[ix)]  $ S_0 (\mathbb{R}^{d}) $ enjoys the projective tensor factorization property:
$ S_0 (\mathbb{R}^{n+m}) = S_0 (\mathbb{R}^{n}) \hat{\otimes} S_0 (\mathbb{R}^{m})$ i.e.
the tensor products of function from $ S_0 (\mathbb{R}^{n}) $ and $ S_0 (\mathbb{R}^{m})$ respectively, given by
$$ f^1 \otimes f^2(x, y) = f^1 (x) \cdot f^2(y), \;\;\; x \in \mathbb{R}^{n} ,y \in \mathbb{R}^{m},
$$
can be used to build any $ f \in S_0 (\mathbb{R}^{n+m})$ by forming absolutely convergent sums, i.e. as
$$
f =
\sum_{k =1}  ^{\infty} f^1 _k \otimes f^2 _k (x, y) \;\;\; \text{with} \;\;\;
\sum_{k =1}  ^{\infty} \| f^1 _k \|_{S_0 (\mathbb{R}^{n}), g^1} \| f^2 _k \|_{S_0 (\mathbb{R}^{m}), g^2} < \infty,
$$
where $ g^1 \in S_0 (\mathbb{R}^{n}) \backslash \left\{0 \right\}$,  $ g^2 \in S_0 (\mathbb{R}^{m}) \backslash \left\{0 \right\}$,
and the corresponding infimum norm provides an equivalent norm.
\end{itemize}
\end{lemma}

\begin{proof} We omit the proof, and refer to the following sources.
{\em i) -- iv)} can be found in \cite{Feichtinger1}, see also \cite[Theorem 3.2.3]{FeiZimm} and \cite[Lemma 14]{FJ2020}.
For {\em v) -- vi)}  we refer to  \cite[Corollary 5.5]{Jakobsen} and \cite[Lemma 3.2.5]{FeiZimm}.
{\em vii)} is  \cite[Proposition 12.1.7]{Grochenig},  \cite[Corollary 4.14]{Jakobsen}, or
\cite[Lemma 15]{FJ2020}, see also  \cite[Sect. 7.1.3]{Gosson3}.
Finally, {\em viii)} and {\em ix)} are proved in \cite[Theorem 7]{Feichtinger1}, see also \cite[Section 9]{Jakobsen}.
\end{proof}

Since different  windows $g \in S_0 (\mathbb{R}^{d}) \backslash \left\{0 \right\}$ give rise to an equivalent norm in
$ S_0 (\mathbb{R}^{d})$ by Lemma \ref{lm:Szero} {\em v)}, from now on instead of $ \| f \|_{S_0, g } $ we shall write
$ \| f \|_{S_0} $.

\begin{definition}\label{definitionGaborProduct}
Let $F,H \in S_0 (\mathbb{R}^{2d}) $ and $g \in S_0 (\mathbb{R}^{d}) \backslash \left\{0 \right\}$.
The {\bf Gabor product} $\natural_{g}$ is defined as
\begin{equation}
\begin{array}{c}
\left(F \natural_{g} H \right) (x, \omega) =
\\
 = \int_{\mathbb{R}^{3d} } \overline{\widehat{g} \left(\omega^{\prime} + \omega^{\prime \prime} -\omega \right)} F(x^{\prime}, \omega^{\prime}) H (x^{\prime}, \omega^{\prime \prime}) e^{2i \pi  x \cdot (\omega^{\prime} + \omega^{\prime \prime} - \omega)} dx^{\prime} d \omega^{\prime} d \omega^{\prime \prime}.
\end{array}
\label{eqWindowedProduct1}
\end{equation}
\end{definition}

\begin{lemma}\label{Lemma1} Let  $F,H \in S_0 (\mathbb{R}^{2d})$ and $g \in S_0 (\mathbb{R}^{d}) \backslash \left\{0 \right\}$.
The Gabor product $ F \natural_{g} H$ given by \eqref{eqWindowedProduct1} is a well-defined product and $\natural_g : S_0 (\mathbb{R}^{2d})\times S_0 (\mathbb{R}^{2d}) \to S_0 (\mathbb{R}^{2d})$ is a continuous mapping.
\end{lemma}

\begin{proof}
We first rewrite \eqref{eqWindowedProduct1} in a more convenient form. The change of variables $\omega^{\prime \prime} \to \xi = \omega^{\prime}+ \omega^{\prime \prime}$ gives
\begin{equation*}
\begin{array}{c}
\left(F \natural_{g} H \right) (x, \omega) = \\
\\
= e^{- 2i \pi \omega \cdot x} \int_{\mathbb{R}^{3d} } \overline{\widehat{g} \left(\xi-\omega \right)} e^{2 i\pi  x \cdot \xi} F(x^{\prime}, \omega^{\prime}) H (x^{\prime}, \xi - \omega^{\prime})  dx^{\prime} d \omega^{\prime} d \xi.
\end{array}
\label{eqWindowedProduct3}
\end{equation*}
Let
\begin{equation}
\begin{array}{c}
\widehat{A}_{F,H} (\xi) =
\int_{\mathbb{R}^{2d}} F(x^{\prime}, \omega^{\prime}) H(x^{\prime}, \xi- \omega^{\prime}) dx^{\prime} d \omega^{\prime}
= \langle F, \overline{ \mathcal{Z}_{\xi} H}  \rangle_{L^2 (\mathbb{R}^{2d})},
\end{array}
\label{eqWindowedProduct5}
\end{equation}
where the operator
\begin{equation}
\left(\mathcal{Z}_{\xi} H \right) (x, \omega) = H (x, \xi-\omega)
\label{eqWindowedProduct6}
\end{equation}
amounts to a reflection and a translation with respect to the second variable.
By Fubini's Theorem, we get
\begin{equation}
\begin{array}{c}
\left(F \natural_{g} H \right) (x, \omega)
= e^{- 2i \pi \omega \cdot x} \int_{\mathbb{R}^d } \overline{\widehat{g} \left(\xi-\omega \right)} e^{2i \pi  x \cdot \xi} \widehat{A}_{F,H} (\xi) d \xi = \\
\\
=e^{- 2i \pi \omega \cdot x}  V_{\widehat{g}} \widehat{A}_{F,H} (\omega, -x) = V_{g} A_{F,H} (x, \omega),
\end{array}
\label{eqWindowedProduct4}
\end{equation}
where we used (\ref{eq2.1}) and (\ref{eq2.2}).

Let us prove that $\widehat{A}_{F,H} \in S_0 (\mathbb{R}^{d}) $.

First assume that $F$ and $G$ are "simple tensors", $ F(x,\omega) = f_1 (x) f_2 (\omega),$ and
$ H(x,\omega) = h_1 (x) h_2 (\omega),$  with $f_1, f_2, h_1, h_2 \in  S_0 (\mathbb{R}^{d}) $. Then by Lemma \ref{lm:Szero} {\em vii)}
it follows that $ g_1 := f_1\cdot h_1 \in S_0 (\mathbb{R}^{d}) $, and
$ g_2  := f_2 \star h_2 \in S_0 (\mathbb{R}^{d}) $. Then \eqref{eqWindowedProduct5} becomes
$$
\widehat{A}_{F,H} (\xi) = \int_{\mathbb{R}^d } g_1 (x') dx' \cdot g_2 (\xi),
$$
so that
\begin{equation}
\begin{array}{c}
\| \widehat{A}_{F,H} (\xi) \|_{S_0}
= \|  (\int_{\mathbb{R}^d } g_1 (x') dx' ) g_2\|_{S_0}  \leq \\
\\
\leq \| g_1 \|_{L^1} \|  g_2\|_{S_0}  \leq \| g_1 \|_{S_0} \|  g_2\|_{S_0}  < \infty,
 \end{array}
\label{eq:hatAinSzero}
\end{equation}
where we used  Lemma \ref{lm:Szero} {\em ii)}.

Similar arguments, together with Lemma \ref{lm:Szero} {\em ix)} can be used to prove that
$\widehat{A}_{F,H} \in S_0 (\mathbb{R}^{d}) $ when $F$ and $H$ are represented as infinite sums of simple tensors instead.

Now $ A_{F,H} \in  S_0 (\mathbb{R}^{d})$ since $ S_0 (\mathbb{R}^{d}) $ is Fourier transform invariant by Lemma \ref{lm:Szero} {\em iv)}, and finally $ F \natural_{g} H \in S_0 (\mathbb{R}^{2d})$ follows from Lemma \ref{lm:Szero} {\em v)}.
The proof is finished.
\end{proof}

Next we prove the main result concerning the Gabor product. 

\begin{theorem}\label{theoremWindowedProduct1}
Let $f,h \in S_0 (\mathbb{R}^d)$ and let $g_1,g_2,g \in S_0 (\mathbb{R}^{d}) \backslash \left\{0 \right\}$.
Then the following identity holds
\begin{equation}
\langle g_2,g_1 \rangle_{L^2 (\mathbb{R}^d)} V_{g} (f \cdot h) =   V_{g_1} (f) \natural_{g} V_{\overline{g_2}} (h),
\label{eqWindowedProduct2}
\end{equation}
and $ V_{g} (f \cdot h)  \in S_0 (\mathbb{R}^{2d})$.
\end{theorem}

\begin{proof}
That $ V_{g} (f \cdot h)  \in  S_0 (\mathbb{R}^{2d})$ follows from  Lemma \ref{lm:Szero} {\em vii)} and {\em v)}.
It remains to prove  \eqref{eqWindowedProduct2}.

By \eqref{eqWindowedProduct4} we have
$$
\left(F \natural_{g} H \right) (x, \omega) = e^{- 2i \pi \omega \cdot x}  V_{\widehat{g}} \widehat{A}_{F,H} (\omega, -x),
$$
where $ \widehat{A}_{F,H}$ is given by \eqref{eqWindowedProduct5}.
From (\ref{eq2.2}) and (\ref{eqFourierTransformConvolution2}) we have
\begin{equation*}
V_{g} (f \cdot h) (x, \omega) = e^{-2 i \pi x \cdot \omega} V_{\widehat{g}} \left( \widehat{ (f \cdot h) } \right) (\omega, - x)= e^{-2 i \pi x \cdot \omega} V_{\widehat{g}} \left( \widehat{f} \star \widehat{h}   \right) (\omega, - x).
\label{eqWindowedProduct7}
\end{equation*}
By comparison with (\ref{eqWindowedProduct4}) it remains to prove that
\begin{equation}
\langle g_2,g_1 \rangle_{L^2 (\mathbb{R}^d)} \left( \widehat{f} \star \widehat{h}   \right) (\xi) =   \widehat{A}_{F,H} (\xi),
\label{eqWindowedProduct8}
\end{equation}
for $F=V_{g_1} f$ and $H=V_{\overline{g_2}} h$.
From (\ref{eqConvolutionInnerProduct}), (\ref{eqFourierTransformProperties}),
Parseval's identity and the orthogonality relations (\ref{eq2.3}), we obtain
\begin{equation}
\begin{array}{c}
\langle g_2,g_1 \rangle_{L^2 (\mathbb{R}^d)} \left( \widehat{f} \star \widehat{h}   \right) (\xi) = \langle g_2,g_1 \rangle_{L^2 (\mathbb{R}^d)} \langle \widehat{f}, \overline{T_{\xi} \mathcal{I} \widehat{h}} \rangle_{L^2 (\mathbb{R}^d)}  = \\
\\
= \langle g_2,g_1 \rangle_{L^2 (\mathbb{R}^d)} \langle  \widehat{f}, T_{\xi} \left(\overline{h} \right)^{ \widehat{}} \;\; \rangle_{L^2 (\mathbb{R}^d)} =  \langle g_2,g_1 \rangle_{L^2 (\mathbb{R}^d)} \langle  \widehat{f},  \left(M_{\xi} \overline{h} \right)^{ \widehat{}} \;\; \rangle_{L^2 (\mathbb{R}^d)} = \\
\\
 =\langle g_2,g_1 \rangle_{L^2 (\mathbb{R}^d)} \langle  f,  M_{\xi} \overline{h}  \rangle_{L^2 (\mathbb{R}^d)} =
 \langle V_{g_1} f, V_{g_2} \left(M_{\xi} \overline{h} \right) \rangle_{L^2 (\mathbb{R}^{2d})}.
\end{array}
\label{eqWindowedProduct10}
\end{equation}
From (\ref{eq2.4}) and (\ref{eq2.5.1}) we have
\begin{equation}
\begin{array}{c}
V_{g_2} \left(M_{\xi} \overline{h} \right) (x, \omega) = V_{g_2} \overline{h} (x, \omega - \xi) = \\
\\
= \overline{\left(V_{\overline{g_2}} h \right) (x, \xi - \omega)} = \mathcal{Z}_{\xi} \left( \overline{V_{\overline{g_2}} h} \right) (x, \omega).
\end{array}
\label{eqWindowedProduct11}
\end{equation}
Finally, if we write $F=V_{g_1} f$ and $H=V_{\overline{g_2}}h$ and substitute (\ref{eqWindowedProduct11}) in (\ref{eqWindowedProduct10}) we recover (\ref{eqWindowedProduct8}).
\end{proof}

We have a straightforward corollary as follows.

\begin{corollary}\label{corollaryHomomorphism}
Let $f,h \in S_0 (\mathbb{R}^d)$ and let $g_1,g_2,g \in S_0 (\mathbb{R}^{d}) \backslash \left\{0 \right\}$ be
such that $g_1=g_2=g$, with $g$ real and $||g||_{L^2 (\mathbb{R}^d)}=1$. Then
\begin{equation}
V_g (f \cdot h) = V_g (f) \natural_g V_g (h).
\label{eqHomomorphism1}
\end{equation}
\end{corollary}

\begin{remark}
We note that in Theorem \ref{theoremWindowedProduct1} and its Corollary \ref{corollaryHomomorphism}
one could take $h$ from the space of multipliers of $S_0$, i.e. the Wiener amalgam space $ W(\mathcal{F} L^1, l^\infty)$. However,
we do not need this generality in the sequel.
\end{remark}

Next we define the product related to the Wigner wave-packet transform.

\begin{definition}\label{definitionWignerProduct1}
Let $F,H \in S_0 (\mathbb{R}^{2d})$ and $g \in S_0 (\mathbb{R}^{d}) \backslash \left\{0 \right\}$.  We define  {\bf the Wigner product} as follows
\begin{equation}
\begin{array}{c}
(F \sharp_{g} H) (x, \omega) =  \int_{\mathbb{R}^{3d}} \overline{\widehat{g} (\omega- \omega^{\prime}-\omega^{\prime \prime})} \times \\
 \\
 \times
   F(x^{\prime},\omega^{\prime}) H(x^{\prime},\omega^{\prime \prime}) e^{i \pi \left[x \cdot \omega - x^{\prime} \cdot(\omega^{\prime} + \omega^{\prime \prime})\right]} e^{2 i \pi x \cdot (\omega^{\prime} + \omega^{\prime \prime} - \omega)} dx^{\prime} d \omega^{\prime} d \omega^{\prime \prime}.
\end{array}
\label{eqWignerProduct1}
\end{equation}
\end{definition}

By \eqref{eqRelationSTFTCrossWignerFunction} and a comparison of the products $ \natural_g $ and $\sharp_{g}$ we conclude that
$ F \sharp_{g} H $ is a well defined map $\sharp_g:S_0 (\mathbb{R}^{2d}) \times S_0(\mathbb{R}^{2d}) \to S_0 (\mathbb{R}^{2d})$.
Moreover, the following is true.

\begin{theorem}\label{TheoremWignerProduct2}
Let $f,h \in S_0 (\mathbb{R}^{d})$ and let $g_1,g_2,g \in S_0 (\mathbb{R}^{d}) \backslash \left\{0 \right\}$. Then
\begin{equation}
\langle g_2,g_1 \rangle_{L^2 (\mathbb{R}^d)} W_{g} (f \cdot h) =     W_{ g_1} (f) \sharp_{g}W_{ \overline{g_2}} (h).
\label{eqWignerProduct2}
\end{equation}
\end{theorem}

\begin{proof}
The result is a simple consequence of Theorem \ref{theoremWindowedProduct1} and (\ref{eqRelationSTFTCrossWignerFunction}).
\end{proof}

Again, an interesting particular case is

\begin{corollary}\label{CorollaryWignerProduct3}
Let the conditions of Theorem \ref{TheoremWignerProduct2} hold with  $g_1=g_2 =   g$, where $g$ is a real valued function such that $||g||_{L^2 (\mathbb{R}^d)}=1$. Then
we have
\begin{equation*}
W_g (f \cdot h) = W_g (f) \sharp_g W_g (h).
\label{eqWignerProduct3}
\end{equation*}
\end{corollary}

\begin{remark} Let $f,h \in S_0(\mathbb{R}^{d})$, and let $g_1,g_2,g \in S_0 (\mathbb{R}^{d}) \backslash \left\{0 \right\}$. By using
\eqref{eqDilationOperator}, \eqref{eqWavepacketTransform}, and \eqref{eqWignerProduct2}
we obtain the following formula for the cross-Wigner transform of the product $f\cdot h$:
\begin{equation*}
W (f \cdot h, g) (x, \omega) =  \langle g_2,g_1  \rangle_{L^2 (\mathbb{R}^d)}^{-1} D_{\frac{1}{2}} [ D_2 W (f, g_1) \sharp_{g}  D_2 W (h, \overline{g_2})] (x, \omega),
\end{equation*}
and likewise, for the Grossmann-Royer transform we obtain
$$ \displaystyle R_g (fh)(x, \omega) =  2^d \langle  g_2,g_1  \rangle_{L^2 (\mathbb{R}^d)}^{-1} D_{\frac{1}{2}} [ D_2 R_{g_1}  (f) \sharp_{g}  D_2 R_{\overline{g_2}} (h )] (x, \omega).$$
\end{remark}

Let us briefly state the main properties of the windowed products. We define the operators that implement modulations and translations of functions on $\mathbb{R}^{2d}$ as
\begin{equation}
(T_{(u, \eta)} F) (x, \omega) = F(x-u, \omega- \eta), \;\;\; (M_{(\eta,u)} F) (x, \omega) = e^{2 i \pi (x \cdot \eta + \omega \cdot u)} F(x, \omega).
\label{eqWignerProduct5}
\end{equation}

\begin{proposition}\label{PropositionProperties1}
The Gabor and the Wigner products are bilinear, commutative and associative on $ S_0 (\mathbb{R}^{2d})$. Moreover, they have the following behavior under time-frequency shifts (the covariance property):
\begin{equation}
\left(M_{(\eta,-u)} T_{(v,\rho)} F\right)  \natural_{g} \left(M_{(- \eta,-u)} T_{(v,\sigma)} H\right) = M_{(0,-u)} T_{(u,\rho + \sigma)} \left(F \natural_{g} H \right),
\label{eqWignerProduct6}
\end{equation}
and
\begin{equation*}
\begin{array}{c}
\left(M_{(\eta,-u)} T_{(v,\rho)} F\right)  \sharp_{g} \left(M_{(\beta,-u)} T_{(v,\sigma)} H\right) = \\
 = M_{(- \rho,-u-v)} T_{\left(u+\frac{v}{2},\rho + \sigma\right)} M_{\left(0, \frac{v}{2} \right)} \left(F \natural_{g} H \right),
 \end{array}
\label{eqWignerProduct6.1}
\end{equation*}
where $\beta$ is such that $2 \eta + 2 \beta = \rho + \sigma$, $F,H \in S_0 (\mathbb{R}^{2d})$ and $g \in  S_0 (\mathbb{R}^{d}) \backslash \left\{0 \right\}$.
\end{proposition}

\begin{proof}
The first statements are trivial. Let us prove (\ref{eqWignerProduct6}). From (\ref{eqWindowedProduct1}) and (\ref{eqWignerProduct5}) we have:
\begin{equation*}
\begin{array}{c}
\left(M_{(\eta,-u)} T_{(v,\rho)} F\right)  \natural_{g} \left(M_{(- \eta,-u)} T_{(v,\sigma)} H\right) (x, \omega)=\\
\\
=  \int_{\mathbb{R}^{3d}} \overline{\widehat{g} (\omega^{\prime} + \omega^{\prime \prime} - \omega)} e^{2 i \pi (\eta \cdot x^{\prime} - u \cdot \omega^{\prime})} F(x^{\prime}-v, \omega^{\prime}- \rho) \times \\
 \\
 \times e^{2 i \pi (-\eta \cdot x^{\prime} - u \cdot \omega^{\prime \prime})} H(x^{\prime}-v, \omega^{\prime \prime}- \sigma) e^{2 i \pi x \cdot (\omega^{\prime} + \omega^{\prime \prime} - \omega)} d x^{\prime} d \omega^{\prime} d \omega^{\prime \prime}.
\end{array}
\label{eqWignerProduct7}
\end{equation*}
After performing the change of variables $x_1 = x^{\prime} -v$, $\omega_1 = \omega^{\prime} - \rho$, $\omega_2 = \omega^{\prime \prime} - \sigma$, and simplifying expressions we obtain
\begin{equation*}
\begin{array}{c}
\left(M_{(\eta,-u)} T_{(v,\rho)} F\right)  \natural_{g} \left(M_{(- \eta,-u)} T_{(v,\sigma)} H\right) (x, \omega)=\\
\\
= e^{- 2 i \pi u \cdot \omega} \int_{\mathbb{R}^{3d}} \overline{\widehat{g} \left(\omega_1 + \omega_2 - (\omega- \rho- \sigma) \right)} F(x_1, \omega_1) \times \\
\\
\times H(x_1, \omega_2) e^{2 i \pi (x-u) \cdot \left(\omega_1 + \omega_2 -(\omega -\rho - \sigma) \right)} dx_1 d \omega_1 d \omega_2 =\\
\\
= e^{- 2 i \pi u \cdot \omega} \left( F \natural_{g} H \right) (x-u, \omega - \rho - \sigma),
 \end{array}
\label{eqWignerProduct8}
\end{equation*}
which terminates the proof.

The proof for the product $\sharp_{g}$ follows in a similar fashion.
\end{proof}

\subsection{Modulation spaces} \label{subs:modsp}

The appropriate functional analysis framework for the study of the STFT is given by modulation spaces introduced in \cite{Feichtinger2}. Their role in time-frequency analysis  is explained in \cite{Grochenig}. We refer to the recent monographs \cite{Benyi, Cordero1} where their use in microlocal analysis, theory of pseudo-differential and Fourier integral operators, and Schr\"odinger equations is highlighted. Since we are interested in weighted modulation spaces, we start with a brief review of weight functions.

\par

A weight in $\mathbb{R}^d$ is a positive function $w_0\in L^{\infty}_{loc}(\mathbb{R}^d)$ such that  $1/w_0\in L^{\infty}_{loc}(\mathbb{R}^d)$. An even weight $w_0$ is {\em submultiplicative} if
\begin{equation}
w_0 (x+y)\leq w_0(x) w_0(y).
\label{eq:submulti}
\end{equation}
A weight $w_0$ on $\mathbb{R}^d$ is {\em  moderate} if there is a submultiplicative weight  $v$ on $\mathbb{R}^d$ such that
\begin{equation}
w_0 (x+y) \precsim v(x) w_0(y),\qquad\forall\,x,y\in\mathbb{R}^d.
\label{eq:moderate}
\end{equation}
In such case we say that $w_0$ is $v-$moderate.

The set of all moderate weights on $\mathbb{R}^d$ is denoted by $ \mathscr{P}_E(\mathbb{R}^d).$
Notice that if  $ w \in  \mathscr{P}_E (\mathbb{R}^d)$ then there is a constant $ r>0$ such that
\begin{equation}
w(x+y) \precsim w (x) e^{r|y|}, \qquad x,y \in \mathbb{R}^d
\label{eq:atmostexp}
\end{equation}
(see e.g. \cite[Lemma 4.2]{Grochenig2}), so that
 $\mathscr{P}_E$ contains weights of at most exponential growth. In particular,
the weights of \textit{polynomial type}, i.e. weights moderate with respect to some polynomial, belong to  $\mathscr{P}_E$ .
This type of weights is sufficient when considering tempered distributions.
However, the space of tempered ultra-distributions $ ({\mathcal S}^{(1)})^{\prime} (\mathbb{R}^d ) $ is
convenient choice when dealing with objects of an exponential type growth.
We refer to \cite{Grochenig2} for a survey on the role of different properties of weight functions in time-frequency analysis.

\par

If  $ w \in  \mathscr{P}_E (\mathbb{R}^{2d})$ then the weighted mixed-norm Lebesgue spaces $L^{p,q} _{w} (\mathbb{R}^{2d})$, $ 1 \leq p,q \leq \infty $,
consist of all $F \in (\mathcal S^{(1)})^{\prime} (\mathbb{R}^{2d})$ such that $ F w \in L^{p,q} (\mathbb{R}^{2d})$, and
$ ||F||_{L^{p,q} _{w}} :=  ||F w ||_{L^{p,q}} $. If $w_{t,s}\in  \mathscr{P}_E (\mathbb{R}^{2d}) $ is a polynomial type weight of the form
\begin{equation}
w_{t,s} (x,\omega) = \langle x \rangle ^t \langle \omega \rangle^s = (1+ |x|^2)^{t/2}  (1+ |\omega|^2)^{s/2},
\label{eqpolynomweight}
\end{equation}
for some $s,t \in \mathbb{R} $, then we use the abbreviated notation $ L^{p,q} _{t,s} (\mathbb{R}^{2d})$.

Note that $ w_{t,s} $ is not submultiplicative when $t,s>0,$ but it is equivalent to the submultiplicative weight
$ (1+|x|)^t (1+|\omega|)^s $, i.e.
$$ w_{t,s} \asymp (1+|x|)^t (1+|\omega|)^s,  $$
and $  L^{p,q} _{t,s} (\mathbb{R}^{2d}) =  L^{p,q} _{(1+|x|)^t (1+|\omega|)^s} (\mathbb{R}^{2d}).$
In many cases it is also convenient to use the rotation invariant weights of the form
$$
w (x,\omega) = \langle (x,\omega)  \rangle ^s = (1+ |x|^2 + |\omega|^2)^{s/2},  \qquad  s\geq 0.
$$

Next we introduce modulation spaces.

\begin{definition}\label{DefinitionModulationSpace1}
Let $g \in \mathcal S^{(1)} (\mathbb{R}^d) \backslash \left\{0 \right\}$, $w \in  \mathscr{P}_E (\mathbb{R}^{2d})$ and $p,q \in \left[1, \infty \right]$ be fixed. The {\bf modulation space} $M_{w}^{p,q} (\mathbb{R}^d)$ consists of all
$f \in (\mathcal S^{(1)})^{\prime} (\mathbb{R}^d)$ such that
\begin{equation}
||f||_{M_{w}^{p,q}} := ||V_g f ||_{L_{w}^{p,q}} < \infty.
\label{eqModulationSpace1}
\end{equation}
\end{definition}

For convenience we set
$ M^{p,p} _w (\mathbb{R}^d) = M^{p} _w (\mathbb{R}^d) $, $ M^{p,q} _1 (\mathbb{R}^d) = M^{p,q} (\mathbb{R}^d) $,
and $M_{w}^{p,q} (\mathbb{R}^d) = M_{t,s}^{p,q} (\mathbb{R}^d)$ if the weight $ w = w_{t,s} $ is given by \eqref{eqpolynomweight}.

\par

It can be proved that modulation spaces are Banach spaces with the norm given by \eqref{eqModulationSpace1}, and
if $w$ is $v-$moderate then Definition \ref{DefinitionModulationSpace1}
is independent of the choice of windows $ g \in M^1 _v (\mathbb{R}^d) $, as different windows lead to equivalent norms
\cite{Feichtinger2, Grochenig, Toft5}. Moreover, if $ p_1\leq p_2$, $q_1\leq q_2$ and $w_2\precsim  w_1 $, $ w_1, w_2 \in \mathscr{P}_E (\mathbb{R}^{2d})$,
then the following continuous embeddings hold
\begin{equation} \label{embeddings}
 \mathcal S^{(1)} (\mathbb{R}^d) \hookrightarrow M^{p_1,q_1}_{w_1}(\mathbb{R}^d)\hookrightarrow M^{p_2,q_2}_{w_2}(\mathbb{R}^d) \hookrightarrow  (\mathcal S^{(1)})^{\prime} (\mathbb{R}^d).
\end{equation}
If, in addition $ p_1, q_1 < \infty,$ then $ \mathcal S^{(1)} (\mathbb{R}^d) $ is dense in  $ M^{p_1,q_1}_{w_1}(\mathbb{R}^d)$.

We also recall that $ \mathcal S (\mathbb{R}^d)$ is dense in $M^{p,q}_{s,s} (\mathbb{R}^d)$ if $ s\geq 0$ and  $1\leq p,q < \infty$.
Moreover, for  $1\leq p,q \leq \infty$ we have
$$
 \mathcal S (\mathbb{R}^d) = \bigcap_{s \geq 0} M^{p,q}_{s,s}(\mathbb{R}^d) \qquad \text{and} \qquad
 \bigcup_{s \geq 0} M^{p,q}_{-s,-s}(\mathbb{R}^d)  = \mathcal S^{\prime}  (\mathbb{R}^d).
$$
Therefore, for the  tempered distributions framework, it is sufficient to consider weights of polynomial growth.
More general weights $ w \in  \mathscr{P}_E (\mathbb{R}^{2d})$ are used  in the study of objects of
(sub)exponential growth/decay at infinity, cf. \cite{To11}. For example,
the Navier-Stokes equation is considered in the context of modulation spaces with exponentially decaying weights in a recent contribution \cite{Feichtinger3}. To emphasize situations when $M_{w}^{p,q} (\mathbb{R}^d)$ contain ultradistributions,
they are sometimes called ultra-modulation spaces, see \cite{Teofanov2, Teofanov3}.

\begin{example}\label{ExampleSzero}
When $ p=q=w=1$,  $M^{1} (\mathbb{R}^d) = S_0 (\mathbb{R}^d)$. In fact, the Feichtinger algebra is the most prominent example of
a modulation space. We will also use the weighted Feichtinger algebra $M^{1} _w (\mathbb{R}^d)$, $w \in \mathscr{P}_E (\mathbb{R}^{2d})$, in Subsection \ref{subsec:extension}.
In particular,  $M^{1} _{0,2} (\mathbb{R}^d)$ will be used in Section \ref{sec:3}.
We will also use the following fact. If   $1\leq p,q < \infty$ and if $w $ is a $v-$ moderate weight, then
by \eqref{embeddings} and the density of  $ \mathcal S^{(1)} (\mathbb{R}^d) $
in  $ M^{p,q}_{w}(\mathbb{R}^d)$ and in  $ L^{p,q}_{w}(\mathbb{R}^d)$ it follows that
$M^{1} _v (\mathbb{R}^d)$ is dense in   $ M^{p,q}_{w}(\mathbb{R}^d)$ and in $ L^{p,q}_{w}(\mathbb{R}^d)$.
\end{example}

\begin{example}\label{ExampleModulationSpace1}
Familiar examples arise when $ p=q=2$. Then $M_{0,0}^{2,2} (\mathbb{R}^d) =M^{2} (\mathbb{R}^d)  = L^2 (\mathbb{R}^d)$, and it can be shown that
\begin{equation*}
M_{0,s}^{2,2} (\mathbb{R}^d) = H^s  (\mathbb{R}^d),  \qquad s \in \mathbb{R},
\label{eqModulationSpace1.1}
\end{equation*}
where $H^s (\mathbb{R}^d)$ is the Sobolev space (also known as the Bessel potential space) of distributions
$f \in \mathcal S^{\prime}(\mathbb{R}^d) $ such that
\begin{equation*}
||f||_{H^s}^2 := \int_{\mathbb{R}^d} (1+ | \omega|^2 )^s | \widehat{f} (\omega)|^2 d \omega < \infty,
\label{eqModulationSpace1.2}
\end{equation*}
cf. \cite[Proposition 11.3.1]{Grochenig}. Furthermore,
if $v_{s} (x,\omega) =  (1+ |x|^2 + |\omega|^2)^{s/2},$ then
$ M_{v_{s}} ^{2}  (\mathbb{R}^d) = Q_s  (\mathbb{R}^d)$, $s \in \mathbb{R},$
where $ Q_s $ denotes the Shubin-Sobolev space, \cite[Lemma 2.3]{Boggiatto}. The spaces $ Q_s $ were introduced in \cite{Shubin1}
to study non-local effects of pseudodifferential operators.
\end{example}

We proceed with several properties and remarks which will be useful when extending  Gabor and Wigner products
in Subsection \ref{subsec:extension}.

By \cite[Proposition 11.3.2]{Grochenig}, we have
\begin{equation}
I_{M_{w}^{p,q}}= \langle \phi,g\rangle_{L^2 (\mathbb{R}^d)}^{-1} V_{\phi}^{\ast} V_{g}~,
\label{eqPre4}
\end{equation}
for all windows $\phi,g \in S_0 (\mathbb{R}^d) \backslash \left\{0 \right\}$ such that
$ \langle \phi,g\rangle_{L^2 (\mathbb{R}^d)} \neq 0$. Here $V_{\phi}^{\ast}$ denotes the adjoint of $V_{\phi}$, defined as
\begin{equation}
V_{\phi}^{\ast}F= \int \int_{\mathbb{R}^{2d}} F(x,\omega) M_{\omega} T_x \phi dx d\omega~,
\label{eqPre5}
\end{equation}
for some $F \in L_{w}^{p,q} (\mathbb{R}^{2d})$. This integral is to be interpreted weakly as:
\[
\begin{array}{c}
<V_{\phi}^{\ast}F,f>
= \int \int_{\mathbb{R}^{2d}} F(x,\omega) <M_{\omega} T_x \phi,f > dx d \omega \\
\\
=\int \int_{\mathbb{R}^{2d}} F(x,\omega)\overline{V_{\phi}f (x, \omega)} dx d \omega ~,
\end{array}
\]
cf. \cite[Definition 11.3.2]{Grochenig}. Moreover, if  $w$ is a $v-$moderate weight, then
\begin{equation}
\|V_{\phi}^{\ast}F\|_{M_{w}^{p,q}} \precsim \|V_{g} \phi \|_{L_{v} ^1} \|F\|_{L_{w}^{p,q}}~,
\label{eqPre7}
\end{equation}
where $\phi \in M^1 _v (\mathbb{R}^d )\backslash \left\{0 \right\}$, and $g_0 \in  M^1 _v (\mathbb{R}^d) \backslash \left\{0 \right\}$ is some reference fixed window.

\begin{remark}\label{Remark1}
Let $V_g \left( S_0 (\mathbb{R}^d)\right)$ denote the range of $V_g(f)$ for a fixed window $g \in S_0 (\mathbb{R}^d)$,
and all $f \in S_0 (\mathbb{R}^d)$. Notice that $ V_g ( S_0 (\mathbb{R}^{d})) \subsetneq S_0 (\mathbb{R}^{2d})$.
In general, $V_{g_1} \left(S_0 (\mathbb{R}^d)\right) \neq V_{g_2} \left( S_0(\mathbb{R}^d)\right)$ when $g_1 \neq g_2$. This can be easily seen from the orthogonality relation, Theorem \ref{theorem2.1}. Namely, if  $\langle g_1, g_2 \rangle_{L^2(\mathbb{R}^d)}=0$, then $\langle V_{g_1}(f_1), V_{g_2}(f_2)\rangle_{L^2(\mathbb{R}^{2d})}=0$, for all $f_1,f_2 \in L^2(\mathbb{R}^d)$.
In addition, $V_{g_1} \left(S_0 (\mathbb{R}^d)\right) = V_{g_2} \left( S_0(\mathbb{R}^d)\right)$ if and only if $g_1 = \alpha g_2,
$ for some $ \alpha \in \mathbb{C} \setminus \{ 0 \}.$
\end{remark}

Next we consider the range of the STFT.

\begin{theorem}\label{TheoremDensityWindows}
Let $1\leq p,q < \infty,$ $ w \in \mathscr{P}_E (\mathbb{R}^{2d})$ be a $v-$moderate weight, and let
$\mathscr{G}= \left\{g_n,~ n \in \mathbb{N}\right\}  $ $  \subset  M^1 _v (\mathbb{R}^d) $
be an orthonormal basis of $ L^2 (\mathbb{R}^d).$
Moreover, let  $\mathcal{V}_{\mathscr{G},w}^{(N),p,q}(\mathbb{R}^{2d})$ denote the closure of the set
\begin{equation}
V_{\mathscr{G}}^{(N)} (\mathbb{R}^{2d})=\left\{\sum_{n=1}^N V_{g_n}(f_n)~:~f_n \in  M^1 _v (\mathbb{R}^d) \right\}~
\label{eqTheoremDensityWindows1}
\end{equation}
with respect to the $L_{w}^{p,q} (\mathbb{R}^{2d})$ norm.
In particular, if $N=1$, we write simply $\mathcal{V}_{g, \omega}^{p,q} (\mathbb{R}^{2d})$ for the closure of $\left\{V_g(f)~: ~f \in M^1 _v  (\mathbb{R}^d) \right\}$.

Then every element of $\mathcal{V}_{\mathscr{G},w}^{(N),p,q}(\mathbb{R}^{2d})$ can be written in the form
\begin{equation}
\sum_{n=1}^N V_{g_n}(f_n)~,
\label{eqTheoremDensityWindows2}
\end{equation}
for some  $f_n \in M_{w}^{p,q}(\mathbb{R}^{d})$, $ n =1,2,\dots, N$.

\end{theorem}

\begin{proof}
Let $F \in \mathcal{V}_{\mathscr{G},w}^{(N),p,q}(\mathbb{R}^{2d})$. Then there exists a sequence
\begin{equation}
H_k= \sum_{n=1}^N V_{g_n}(f_n^{(k)})~, ~ f_n^{(k)} \in M^1 _v (\mathbb{R}^d)~,
\label{eqTheoremDensityWindows3}
\end{equation}
such that
\begin{equation}
\|H_k -F \|_{L_{w}^{p,q} (\mathbb{R}^{2d})} \to 0~, ~\text{ as } k \to \infty~.
\label{eqTheoremDensityWindows4}
\end{equation}
Since $(H_k)$ is a Cauchy sequence, we also have:
\begin{equation}
\|H_k -H_l \|_{L_{w}^{p,q} (\mathbb{R}^{2d})} \to 0~, ~\text{ as } k,l \to \infty~.
\label{eqTheoremDensityWindows5}
\end{equation}
On the other hand, if $g_0 \in M^1 _v  (\mathbb{R}^d)$ is a fixed window, then for each fixed $m=1,2,3,\cdots,N$, using \eqref{eqPre4} and \eqref{eqPre7} we obtain:
\begin{equation}
\begin{array}{c}
\|H_k -H_l \|_{L_{w}^{p,q} (\mathbb{R}^{2d})}=\|\sum_{n=1}^N V_{g_n}(f_n^{(k)})-\sum_{n=1}^N V_{g_n}(f_n^{(l)}) \|_{L_{w}^{p,q} (\mathbb{R}^{2d})}=\\
\\
=\|\sum_{n=1}^N V_{g_n}\left(f_n^{(k)}-f_n^{(l)} \right) \|_{L_{w}^{p,q} (\mathbb{R}^{2d})}\\
\\
 \gtrsim \frac{1}{\|V_{g_0}g_m\|_{L_v^1}} \|V_{g_m}^{\ast}\left(\sum_{n=1}^N V_{g_n}\left(f_n^{(k)}-f_n^{(l)} \right)\right)\|_{M_{w}^{p,q} (\mathbb{R}^{d})}=\\
\\
= \frac{1}{\|V_{g_0}g_m\|_{L_v^1}} \|\sum_{n=1}^N V_{g_m}^{\ast} V_{g_n}\left(f_n^{(k)}-f_n^{(l)} \right)\|_{M_{w}^{p,q} (\mathbb{R}^{d})}=\\
\\
=\frac{1}{\|V_{g_0}g_m\|_{L_v^1}} \|f_m^{(k)}-f_m^{(l)}  \|_{M_{w}^{p,q} (\mathbb{R}^{d})}~.
\end{array}
\label{eqTheoremDensityWindows6}
\end{equation}
Thus, $\left(f_m^{(k)}\right)$ is a Cauchy sequence in $M_{w}^{p,q} (\mathbb{R}^{d})$. Since $M_{w}^{p,q} (\mathbb{R}^{d})$ is complete, there exists $f_m \in M_{w}^{p,q} (\mathbb{R}^{d})$, such that
\begin{equation}
\|f_m^{(k)}- f_m \|_{M_{w}^{p,q} (\mathbb{R}^{d})} \to 0~, ~\text{ as } k \to \infty~.
\label{eqTheoremDensityWindows7}
\end{equation}
Finally, for every $\epsilon >0$ and every $m=1,2, \cdots, N$, there exists $K \in \mathbb{N}$, such that
\begin{equation}
\|f_m^{(k)}-f_m\|_{M_{w}^{p,q} (\mathbb{R}^{d})}  < \frac{\epsilon}{N}~,
\label{eqTheoremDensityWindows8}
\end{equation}
whenever $k \geq K$.

It follows that:
\begin{equation}
\begin{array}{c}
\|\sum_{n=1}^N V_{g_n}(f_n^{(k)})-\sum_{n=1}^N V_{g_n}(f_n)\|_{L_{w}^{p,q} (\mathbb{R}^{2d})} =\\
\\
=\|\sum_{n=1}^N V_{g_n}\left(f_n^{(k)}-f_n \right)\|_{L_{w}^{p,q} (\mathbb{R}^{2d})} \\
\\
\leq \sum_{n=1}^N \|V_{g_n}\left(f_n^{(k)}-f_n\right) \|_{L_{w}^{p,q} (\mathbb{R}^{2d})} \asymp N \|f_n^{(k)}-f_n \|_{M_{w}^{p,q} (\mathbb{R}^{2d})} < \epsilon ~.
\end{array}
\label{eqTheoremDensityWindows9}
\end{equation}
This means that $\sum_{n=1}^N V_{g_n}(f_n^{(k)}) \to \sum_{n=1}^N V_{g_n}(f_n)$, as $k \to \infty$, and hence $F=\sum_{n=1}^N V_{g_n}(f_n)$.
\end{proof}

We end this subsection by proving a density theorem, which might be
considered folklore, but the authors couldn't find its published
version elsewhere.

\begin{theorem}\label{Theorem2}
Let $1\leq p,q < \infty,$ let $ w \in \mathscr{P}_E (\mathbb{R}^{2d})$ be a $v-$moderate weight,
and let $\left(f_n\right)_{n \in \mathbb{N}} \in M^{1} _v (\mathbb{R}^d)$ be an orthonormal basis of $L^2 (\mathbb{R}^d)$.
Then the closure of all linear combinations of $\left(V_{f_n}(f_m) \right)_{n,m \in \mathbb{N}}$ is dense in $L_{w}^{p,q}(\mathbb{R}^{2d})$.
\end{theorem}

%
%

\begin{proof}
Let us assume first that $F \in  M^{1} _v (\mathbb{R}^{2d})$, and consider an arbitrary $G\in M^{1} _v (\mathbb{R}^{2d}) $. Define
\begin{equation}
a_{nm}= \langle F, V_{f_n}(f_m)\rangle_{L^2 (\mathbb{R}^{2d})}= \langle V_{f_n}^{\ast} F, f_m \rangle_{L^2 (\mathbb{R}^d)}~,
\label{eqTheorem21}
\end{equation}
where $V_{f_n}^{\ast}$ denotes the adjoint of the STFT, (\ref{eqPre5}).

Recall that $<\cdot,\cdot>$ denotes the duality bracket $(L_{w}^{p,q})^{\prime }(\mathbb{R}^{2d}) \times L_{w}^{p,q}(\mathbb{R}^{2d}) \to \mathbb{C}$,
whereas $\langle \cdot, \cdot \rangle_{L^2(\mathbb{R}^{2d})}$ is the inner product in $ L^2(\mathbb{R}^{2d})$.
We have:
\begin{equation}
\begin{array}{c}
\left| < G, \sum_{n,m=1}^N a_{nm} V_{f_n} (f_m)-F> \right|
=\left| \langle G, \overline{\sum_{n,m=1}^N a_{nm} V_{f_n} (f_m)-F} \rangle_{L^2(\mathbb{R}^{2d})} \right|\\
\\
\leq \|G\|_{L^2 (\mathbb{R}^{2d})} \|\sum_{n,m=1}^N a_{nm} V_{f_n} (f_m)-F\|_{L^2 (\mathbb{R}^{2d})}~.
\end{array}
\label{eqTheorem22}
\end{equation}
Now, since $\left(V_{f_n}(f_m) \right)_{n,m}$ is obviously an orthonormal basis of $L^2(\mathbb{R}^{2d})$, it follows that $F$ can be expanded in that basis with the coefficients $a_{nm} $ given by \eqref{eqTheorem21}.
Then (\ref{eqTheorem22}) implies that
\begin{equation*}
\big | <G, \sum_{n,m=1}^N a_{nm} V_{f_n} (f_m)-F> \big | \to 0~,
\label{eqTheorem24}
\end{equation*}
as $N \to \infty$, for all $G \in  M^{1} _v (\mathbb{R}^{2d})$. Since $ M^{1} _v (\mathbb{R}^{2d})$ is dense in
$(L_{w}^{p,q})^{\prime } (\mathbb{R}^{2d})$, we conclude that
\begin{equation*}
\sum_{n,m=1}^N a_{nm} V_{f_n} (f_m) \rightharpoonup F~,
\label{eqTheorem25}
\end{equation*}
where $\rightharpoonup $ denotes the weak convergence in $L_{w}^{p,q} (\mathbb{R}^{2d})$.

By Mazur's Lemma (cf. \cite{Ekeland, Mazur}), there exists a function $\mathcal{N}: \mathbb{N} \to \mathbb{N}$ and a sequence of sets of real numbers $\left\{\alpha (N)_k,~k=N, \cdots, \mathcal{N} (N) \right\}$, with $\alpha(N)_k \geq 0$ and $\sum_{k=N}^{\mathcal{N}(N)} \alpha (N)_k =1$, such that the sequence
\begin{equation*}
H_N=\sum_{k=N}^{\mathcal{N}(N)} \alpha (N)_k \sum_{n,m=1}^k a_{nm}V_{f_n}(f_m)=
\sum_{n,m=1}^{\mathcal{N}(N)}  b_{nm} V_{f_n}(f_m)~,
\label{eqTheorem26}
\end{equation*}
converges strongly to $F$ in $L_{w}^{p,q} (\mathbb{R}^{2d})$. Here $(b_{n	m}) $ denotes a new set of complex coefficients which, for $1\leq n,m \leq \mathcal{N}(N)$, can be expressed as a linear combinations of the coefficients $\left\{a_{nm},~1\leq n,m \leq \mathcal{N}(N) \right\}$.

Finally, choose an arbitrary $F \in L_{w}^{p,q} (\mathbb{R}^{2d})$. Since $M^{1} _v (\mathbb{R}^{2d})$ is dense in $L_{w}^{p,q} (\mathbb{R}^{2d})$, for any $\epsilon >0$, there exists $F_0 \in M^{1} _v (\mathbb{R}^{2d})$ such that $\|F-F_0\|_{L_{w}^{p,q}} < \frac{\epsilon}{2}$. From the previous result, there exists $N \in \mathbb{N}$ and a set of complex numbers $\left(b_{nm}\right)_{nm}$ such that $\|F_0 - \sum_{n,m=1}^N b_{nm} V_{f_n}(f_m)\|_{L_{w}^{p,q}} < \frac{\epsilon}{2}$ and it finally follows that
$$\|F - \sum_{n,m=1}^N b_{nm} V_{f_n}(f_m)\|_{L_{w}^{p,q}} < \epsilon,$$
which proves the result.
\end{proof}

\subsection{Extension of the products} \label{subsec:extension}

To extend the products from Subsection \ref{subs:products} we use multiplication properties for modulation spaces. Initial general results from \cite{Feichtinger2} based on the Fourier transforms of convolutions in Wiener-amalgam spaces, were thereafter reconsidered by different authors, \cite{Benyi,Cordero1,Toft4,Toft1, Wang}. Here we use a recent result from \cite{Toft3}
since it contains quite general and simple conditions on the weight functions.

In the sequel we shall use the Young functional given by
\begin{equation}
R(p)=1 + \frac{1}{p_0} - \frac{1}{p_1}-\frac{1}{p_2},
\label{eqYoungfunctional}
\end{equation}
for $p=(p_0, p_1, p_2) \in \left[1, \infty \right]^3$.
The common conditions for the H\"older and the Young inequalities are then given by
$ R(p) = 1,$ and $ R(p) = 0,$ respectively.

\par

To extend the Gabor product we first recall a  result on pointwise products in modulation spaces
related to general weights $ w_0, w_1, w_2 \in \mathscr{P}_E(\mathbb{R}^{2d})$ which satisfy the condition
\begin{equation}
w_0 (x, \omega_1 + \omega_2) \precsim w_1 (x, \omega_1 ) w_2 (x, \omega_2).
\label{eqweightcondition}
\end{equation}

\begin{theorem}\label{ThmModSpaceProduct1} \cite[Theorem 3.2]{Toft3}
Let there be given $p_j,q_j \in \left[1, \infty \right]$, $j=0,1,2$, such that $ R(p) \leq 1 $ and $R(q) \leq 0$, and let
$w_0,  w_1, w_2 \in  \mathscr{P}_E (\mathbb{R}^{2d})$ satisfy
\eqref{eqweightcondition}.
Then the map $(f_1,f_2) \mapsto f_1 \cdot f_2 $ on $C_0 ^\infty (\mathbb{R}^d)$ extends to a continuous map from $M_{w_1}^{p_1,q_1} (\mathbb{R}^d) \times M_{w_2}^{p_2,q_2} (\mathbb{R}^d)$ to $M_{w_0}^{p_0,q_0} (\mathbb{R}^d)$. The extension is unique when $p_j,q_j < \infty$, $j=1,2$.
\end{theorem}

If the conditions in Theorem \ref{ThmModSpaceProduct1} hold, then its result can be restated as follows:
\begin{equation*}
||f_1 \cdot f_2||_{M_{w_0}^{p_0,q_0}} \precsim ||f_1||_{M_{w_1}^{p_1,q_1}} ~ ||f_2||_{M_{w_2}^{p_2,q_2}}.
\label{eqModulationSpace8}
\end{equation*}

\par

Note that once the cases $ R(p) = 1$ and $ R(q) = 0$ are proved, then the more general conditions
$ R(p) \leq 1 $ and $R(q) \leq 0$ are provided by the embeddings in \ref{embeddings}.
We also refer to  \cite[Proposition 2.4.23]{Cordero1} for the related result which follows from
Theorem \ref{ThmModSpaceProduct1} since the involved weights satisfy \ref{eqweightcondition}.

\par
Next we consider an extension of the Gabor product related to
$\mathcal{V}_{\mathscr{G},w}^{(N),p,q}(\mathbb{R}^{2d})$.

\par

Recall that, by Theorem \ref{TheoremDensityWindows} all elements of $ \mathcal{V}_{\mathscr{G},w}^{(N),p,q}(\mathbb{R}^{2d})$ are of the form $\sum_{n=1}^N V_{g_n}(f_n)$, with $f_n \in M_{w}^{p,q}(\mathbb{R}^{2d})$.

\begin{theorem}\label{TheoremModulationSpace4}
Let there be given $N,M \in \mathbb{N}$, $1 \leq p_j,q_j < \infty$, $j=0,1,2$, such that $ R(p) \leq 1 $ and $R(q) \leq 0$.
Furthermore, let $ w_j \in \mathscr{P}_E (\mathbb{R}^{2d})$ be $v_j -$moderate and such that \eqref{eqweightcondition} holds,
let  $g \in M^{1} _{v_0} (\mathbb{R}^d) \backslash \left\{0 \right\}$,
and let $\mathscr{G} =\left\{g_j\right\} \subset   M^{1} _{v}(\mathbb{R}^d)$,  $ v= \max \{ v_1, v_2 \},$ be an orthonormal basis of $L^2 (\mathbb{R}^d)$, and denote by $\overline{\mathscr{G}} =\left\{\overline{g}_j\right\}$ the orthonormal basis of the complex-conjugates.

Then the Gabor product $\natural_{g}$ extends from
$V_{\mathscr{G}}^{(N)} \left(\mathbb{R}^{2d}\right) \times
V_{\overline{\mathscr{G}}}^{(M)} \left(\mathbb{R}^{2d}\right) $
to a continuous map from $\mathcal{V}_{\mathscr{G},w_1}^{(N),p_1,q_1}(\mathbb{R}^{2d}) \times \mathcal{V}_{\overline{\mathscr{G}},w_2}^{(M),p_2,q_2}(\mathbb{R}^{2d}) $ to $ \mathcal{V}_{g,w_0}^{p_0,q_0}(\mathbb{R}^{2d})$, and
\begin{equation}
||F_1 \natural_{g} F_2||_{L_{w_0}^{p_0,q_0}} \precsim ||F_1||_{L_{w_1}^{p_1,q_1}} ~ ||F_2||_{L_{w_2}^{p_2,q_2}}~,
\label{eqModulationSpace8.1A}
\end{equation}
for all $F_1 \in \mathcal{V}_{\mathscr{G},w_1}^{(N),p_1,q_1}(\mathbb{R}^{2d})$,  $F_2 \in   \mathcal{V}_{\overline{\mathscr{G}},w_2}^{(M),p_2,q_2}(\mathbb{R}^{2d}) $.
\end{theorem}

We remark that the hidden constants appearing in (\ref{eqModulationSpace8.1A}) may depend on $d$, the indices $p_0,q_0, \cdots$, the weights $w_0,w_1, \cdots$, $s_0,t_0, \cdots$ and on the elements $g_1, \cdots, g_K$ $(K=\text{max} \left\{N,M\right\})$ of the basis, but not on $F_1$ and $F_2$.

\begin{proof}
We aim to prove that
\begin{equation}
\begin{array}{c}
\|\sum_{n=1}^N V_{g_n} (f_n) \natural_g \sum_{m=1}^N V_{\overline{g_m}} (h_m) \|_{L_{w_0}^{p_0,q_0}} \precsim  \\
\precsim \|\sum_{n=1}^N V_{g_n} (f_n)\|_{L_{w_1}^{p_1,q_1}} \|\sum_{m=1}^N V_{\overline{g_m}} (h_m)\|_{L_{w_2}^{p_2,q_2}}~,
\end{array}
\label{eqTheorem32}
\end{equation}
where the weights and the Lebesgue parameters satisfy the conditions of Theorem \ref{ThmModSpaceProduct1}.

\par

Note that we can assume $N=M$. If, say $N>M$, then we can write: $\sum_{m=1}^M  V_{\overline{g_m}}(h_m)=\sum_{m=1}^N  V_{\overline{g_m}}(h_m)$, with $h_{m}=0$, for $m= M+1, \cdots, N$.

Since $(g_n)_n$ is an orthonormal basis, we have from (\ref{eqWindowedProduct2}):
\begin{equation}
V_{g_n}(f_n) \natural_g V_{\overline{g_m}} (h_m)= \delta_{n,m} V_{g} (f_n \cdot h_n)~,
\label{eqTheorem33}
\end{equation}
where $ \delta_{n,m}$ is the Kronecker delta. Consequently:
\begin{equation}
\begin{array}{c}
\sum_{n=1}^N V_{g_n} (f_n) \natural_g \sum_{m=1}^N V_{\overline{g_m}} (h_m)=\sum_{n,m=1}^N V_{g_n} (f_n) \natural_g V_{\overline{g_m}} (h_m)=\\
\\
=\sum_{n=1}^N V_{g} (f_n \cdot h_n)=V_{g} \left(\sum_{n=1}^Nf_n \cdot h_n\right)~.
\end{array}
\label{eqTheorem34}
\end{equation}
We then have:
\begin{equation}
\begin{array}{c}
\|\sum_{n=1}^N V_{g_n} (f_n) \natural_g \sum_{m=1}^N V_{\overline{g_m}} (h_m) \|_{L_{w_0}^{p_0,q_0}}=\\
\\
=\|V_{g} \left(\sum_{n=1}^Nf_n \cdot h_n\right)\|_{L_{w_0}^{p_0,q_0}}=\|\sum_{n=1}^Nf_n \cdot h_n\|_{M_{w_0}^{p_0,q_0}}
\end{array}
\label{eqTheorem35}
\end{equation}

To make the proof clearer, let us start with the case $N=2$.

We can write $f_1 \cdot h_1 +f_2 \cdot h_2$ as
\begin{equation}
f_1 \cdot h_1 +f_2 \cdot h_2= \frac{1}{2} \left((f_1+f_2) \cdot (h_1+h_2)+(f_1-f_2) \cdot(h_1-h_2) \right)~.
\label{eqTheorem37}
\end{equation}

However, in order to consider a generalization to the case $N>2$ it is better to consider the following more redundant decomposition:
\begin{equation}
\begin{array}{c}
f_1 \cdot h_1 +f_2 \cdot h_2= \frac{1}{2^2} \left((f_1+f_2) \cdot (h_1+h_2) +(f_1-f_2) \cdot (h_1-h_2)
+\right.\\
\\
\left.+(-f_1+f_2) \cdot (-h_1+h_2)+(-f_1-f_2) \cdot(-h_1-h_2) \right)~,
\end{array}
\label{eqTheorem37A}
\end{equation}
For $N \geq 2$, we can write:
\begin{equation}
\begin{array}{c}
f_1 \cdot h_1 + \cdots f_N \cdot h_N= \frac{1}{2^N} \sum_{\sigma_1, \cdots, \sigma_N=\pm} \left(\sigma_1 f_1 + \cdots \sigma_N f_N \right) \cdot \left(\sigma_1 h_1 + \cdots \sigma_N h_N \right)=\\
\\
= \frac{1}{2^N} \sum_{\left[\Sigma\right]} \sum_{j=1}^N \sigma_j f_j \sum_{m=1}^N \sigma_m h_m~,
\end{array}
\label{eqTheorem37B}
\end{equation}
where we used the compact notation $\sum_{\sigma_1, \cdots, \sigma_N=\pm}=\sum_{\left[\Sigma\right]}$.

It follows from the triangle inequality, Theorem \ref{ThmModSpaceProduct1} and \eqref{eqPre4}, that
\begin{equation}
\begin{array}{c}
\|\sum_{n=1}^N f_n \cdot h_n  \|_{M_{w_0}^{p_0,q_0}} = \\
\\
=\frac{1}{2^N}\| \sum_{\left[\Sigma\right]}  \left(\sum_{j=1}^N \sigma_j f_j\right)\cdot \left( \sum_{m=1}^N \sigma_m h_m \right) \|_{M_{w_0}^{p_0,q_0}}\\
\\
\leq \frac{1}{2^N} \sum_{\left[\Sigma\right]} \| \left(\sum_{j=1}^N \sigma_j f_j\right)\cdot \left( \sum_{m=1}^N \sigma_m h_m \right) \|_{M_{w_0}^{p_0,q_0}}
\\
\\
\precsim  \frac{1}{2^N }  \sum_{\left[\Sigma\right]}\| \sum_{j=1}^N \sigma_j f_j \|_{M_{w_1}^{p_1,q_1}}  \| \sum_{m=1}^N \sigma_m h_m \|_{M_{w_2}^{p_2,q_2}} \\
\\
 \asymp \frac{1}{2^N}  \sum_{\left[\Sigma\right]}  \|\sum_{j=1}^N \frac{\sigma_j}{ \langle \gamma_{\left[\Sigma\right]} , g_j \rangle_{L^2(\mathbb{R}^d)}} V_{\gamma_{\left[\Sigma\right]}}^{\ast} V_{g_j}(f_j) \|_{M_{w_1}^{p_1,q_1}}  \times \\
 \\
 \times \|\sum_{m=1}^N \frac{\sigma_m}{ \langle \eta_{\left[\Sigma\right]} , \overline{g_m} \rangle_{L^2(\mathbb{R}^d)}} V_{\eta_{\left[\Sigma\right]}}^{\ast} V_{\overline{g_m}}(h_m) \|_{M_{w_2}^{p_2,q_2}}
\end{array}
\label{eqTheorem38}
\end{equation}

The windows $\gamma_{\left[\Sigma\right]},~\eta_{\left[\Sigma\right]} \in M^1 _v$ can be taken are arbitrary, besides the obvious conditions
$$
\langle \gamma_{\left[\Sigma\right]} , g_j \rangle_{L^2(\mathbb{R}^d)} \neq 0, \quad
\langle \eta_{\left[\Sigma\right]} , \overline{g_j} \rangle_{L^2(\mathbb{R}^d)}\neq 0, \quad j= 1,\dots, N.
$$

By choosing
$$
\gamma_{\left[\Sigma\right]} = \sum_{j=1} ^N \sigma_j g_j \;\;\; \text{and} \;\;\;
\eta_{\left[\Sigma\right]} =  \sum_{m=1} ^N \sigma_m \overline{g_m},
$$
we obtain
\begin{equation}
\begin{array}{l}
\frac{\sigma_j}{ \langle \gamma_{\left[\Sigma\right]} , g_j \rangle_{L^2(\mathbb{R}^d)}}= \frac{\sigma_1}{\langle \gamma_{\left[\Sigma\right]} , g_1 \rangle_{L^2(\mathbb{R}^d)}},~ \forall j=1, \cdots, N, \\
\\
\frac{\sigma_m}{ \langle \eta_{\left[\Sigma\right]} , \overline{g_m} \rangle_{L^2(\mathbb{R}^d)}}= \frac{\sigma_1}{\langle \eta_{\left[\Sigma\right]} , \overline{g_1} \rangle_{L^2(\mathbb{R}^d)}},~ \forall m=1, \cdots, N
\end{array}
\label{eqTheorem39}
\end{equation}

From \eqref{eqTheorem38},\eqref{eqTheorem39} and \eqref{eqPre7}, it follows that:
\begin{equation}
\begin{array}{c}
\|\sum_{n=1}^N f_n \cdot h_n  \|_{M_{w_0}^{p_0,q_0}} \\
\\
\precsim  \frac{1}{2^N}  \sum_{\left[\Sigma\right]}\frac{\|V_{g}\gamma_{\left[\Sigma\right]}\|_{L_v^1} \|V_{g}\eta_{\left[\Sigma\right]}\|_{L_v^1}}{ \left|\langle \gamma_{\left[\Sigma\right]} , g_1 \rangle_{L^2(\mathbb{R}^d)}  \langle \eta_{\left[\Sigma\right]} , \overline{g_1} \rangle_{L^2(\mathbb{R}^d)} \right|} \|\sum_{j=1}^N V_{g_j} f_j  \|_{L_{w_1}^{p_1,q_1}}  \|\sum_{m=1}^N V_{\overline{g_m}} h_m  \|_{L_{w_2}^{p_2,q_2}} ~.
\end{array}
\label{eqTheorem310}
\end{equation}

Finally, from \eqref{eqTheorem35} and \eqref{eqTheorem310}, we obtain \eqref{eqTheorem32}, which proves the theorem.
\end{proof}


\begin{remark}
The condition $R(q) \leq 0 $ in Theorem \ref{ThmModSpaceProduct1}
can be relaxed into $R(q)\leq 1/2$ when the involved weights are of a polynomial growth. Then we may use
\cite[Theorem 2.4 (2)]{Toft1} to obtain extensions of the Gabor product similar to the ones given in
Theorem \ref{ThmModSpaceProduct1}. We omit the details, since the conditions on the involved weights then become quite technical,
cf. \cite{Toft1}.
\end{remark}

Theorem \ref{TheoremModulationSpace4}  holds for arbitrary $N,M \in \mathbb{N}$. Using the density theorem (Theorem \ref{Theorem2}), we may be tempted to take the limit $N,M \to \infty$ and thus extend the results to the entire mixed-norm spaces. However, with the technique used in the proof of Theorem \ref{TheoremModulationSpace4}, the upper bound constants may grow unboundedly as $N,M \to \infty$. Nevertheless, we are convinced that the result can be extended through some other method, and we state it as a conjecture.

\begin{conjecture}\label{Conjecture1} Let the assumptions on the Lebesgue parameters and weight functions from Theorem
\ref{TheoremModulationSpace4} hold true.
Then the inequality (\ref{eqModulationSpace8.1A}) holds for all $F_1 \in L_{w_1}^{p_1,q_1}(\mathbb{R}^{2d})$,  $F_2 \in  L_{w_2}^{p_2,q_2}(\mathbb{R}^{2d}) $.
\end{conjecture}

Although we were unable to prove the conjecture in its full generality, we can nevertheless prove that one can extend the product continuously to specific mixed-norm Lebesgue spaces (see Theorem \ref{thm:extension-version1}).

For these extensions we use further properties of modulation spaces. Firstly, we use
embedding relations between Fourier--Lebesgue and modulations spaces:
\begin{equation}
\mathcal{F} L ^{q}  (\mathbb{R}^d) \hookrightarrow M^{ p, q}  (\mathbb{R}^d), \;\;\; 1\leq q \leq \infty, \;\;\;
p \geq \text{max} \{ q, q'  \}.
\label{eq:fourierlebesguesembedding}
\end{equation}
Recall,  $ f \in \mathcal{F} L ^{q}  (\mathbb{R}^d)$ if $ \hat f \in  L ^{q}  (\mathbb{R}^d)$ and
$$
\| f \|_{\mathcal{F} L ^{q} } = \| \hat f\|_{ L ^{q}}.
$$
For the embeddings \eqref{eq:fourierlebesguesembedding} we refer to  \cite[Corollary 1.1]{Cunanan1} and \cite[Proposition 1.7]{Toft4}.
We also note that Fourier--Lebesgue and modulation spaces are locally the same, see \cite[Corollary 2]{Toft2}.

\par

Apart from these embeddings we also use the fact  that the map $f \mapsto  \langle D \rangle ^{s_0} f$  is a homeomorphism from $ M^{p,q} _{t, s+s_0} (\mathbb{R}^d)$
to  $ M^{p,q} _{t, s} (\mathbb{R}^d)$, $ 1 \leq p,q \leq \infty$, $ t,s,s_0 \in \mathbb{R}$,
see e.g. \cite[Theorem 2.3.14]{Cordero1} and \cite[Corollary 2.3]{Toft5}. Here $D =\frac{1}{\sqrt{-1}} \partial$ is the operator of differentiation, and $\langle D \rangle ^{s_0} f = \mathcal{F} \langle \cdot \rangle ^{s_0} \hat f$ is the Fourier multiplier, $s_0 \in \mathbb{R}$.

\begin{theorem}
\label{thm:extension-version1}
Let there be given $ 1\leq p,q_1,q_2 \leq \infty,$ and let $p'$ and $ q_0$ be determined by
$\frac{1}{p} + \frac{1}{p^{\prime}}=1$ and  $ R(q)=0$.
Then the following holds:
\begin{enumerate}
\item If $r \geq \text{max} \{ q_0, q' _0 \},$ and  $g \in S_0 (\mathbb{R}^d) \backslash \left\{0 \right\}$, then the
Gabor product $\natural_{g}$ extends from $S_0 (\mathbb{R}^{2d}) \times S_0 (\mathbb{R}^{2d})$
to a continuous map from $L^{p,q_1}  (\mathbb{R}^{2d}) \times L^{p',q_2}  (\mathbb{R}^{2d})$ to $L^{r,q_0} (\mathbb{R}^{2d})$, and
\begin{equation*}
||F \natural_{g} H||_{L^{r,q_0}} \precsim ||F||_{L^{p,q_1}} ~ ||H||_{L^{p',q_2}}.
\label{eqModulationSpace8.1}
\end{equation*}

\item If $r \geq \text{max} \{ q_0, q' _0 \}$,
$t_1 + t_2 \geq 0,$ $ s_1+s_2 \geq 0,$ $ \text{min} \{s_1, s _2 \} \geq s$,
 and  $g \in M^1 _{0,s} (\mathbb{R}^d) \backslash \left\{0 \right\}$,
then the Gabor product $\natural_{g}$ extends from $M^{1} _{t_1,s_1} (\mathbb{R}^{2d}) \times M^{1} _{t_2,s_2} (\mathbb{R}^{2d})$
to a continuous map from $L_{t_1,s_1} ^{p,q_1}  (\mathbb{R}^{2d}) \times L_{t_2,s_2} ^{p',q_2}  (\mathbb{R}^{2d})$ to
$L _{0,s} ^{r,q_0} (\mathbb{R}^{2d})$, and
\begin{equation*}
||F \natural_{g} H ||_{ L_{0,s}^{r,q_0}} \precsim ||F ||_{L_{t_1,s_1} ^{p,q_1}} ~ || H  ||_{L_{t_2,s_2} ^{p',q_2}}.
\label{eqModulationSpace8.2}
\end{equation*}
\end{enumerate}

\end{theorem}

Note that {\em 1.} follows  from {\em 2.} when $ t_1 = t_2 = s_1 = s_2 = s = 0.$
Therefore, it is enough to show {\em 2.}. However, to emphasize  an extra argument used in the weighted case,
we give  a detailed proof as follows.

\par

\begin{proof}
{\em 1.} We first note that $ \widehat{A}_{F,H} \in L^{q_0} (\mathbb{R}^{d}), $ where $ \widehat{A}_{F,H} $ is given by \eqref{eqWindowedProduct5}. This follows from the applications of H\"older's inequality with respect to the first variable, and Young's inequality with respect to the second variable:
\begin{equation}
\begin{array}{c}
\| \widehat{A}_{F,H} (\xi) \|_{L^{q_0}}
= \int_{\mathbb{R}^{d}}
| \int_{\mathbb{R}^{2d}} F(x^{\prime}, \omega^{\prime}) H(x^{\prime}, \xi- \omega^{\prime}) dx^{\prime} d \omega^{\prime}|^{q_0} d\xi )^{1/q_0}
 \leq \\
\\
\leq
(\int_{\mathbb{R}^{d}} (\int_{\mathbb{R}^{2d}} | F(x^{\prime}, \omega^{\prime}) H(x^{\prime}, \xi- \omega^{\prime})|
dx^{\prime} d \omega^{\prime} )^{q_0} d\xi )^{1/q_0} \leq \\
\\
\leq
(\int_{\mathbb{R}^{d}}
(\int_{\mathbb{R}^{d}}
\| F (\cdot, \omega^{\prime}) \|_{L^p}  \| H (\cdot, \xi - \omega^{\prime}) \|_{L^{p'}} d \omega^{\prime} )^{q_0}
d\xi )^{1/q_0} = \\
\\
=
(\int_{\mathbb{R}^{d}}
(\int_{\mathbb{R}^{d}} \tilde F( \omega^{\prime}) \tilde H ( \xi - \omega^{\prime})  d \omega^{\prime} )^{q_0}
d\xi )^{1/q_0} =\\
\\
=
(\int_{\mathbb{R}^{d}} | (\tilde F * \tilde H) ( \xi)|^{q_0}
d\xi )^{1/q_0}
\precsim \\
\\
\precsim \| \tilde F \|_{L^{q_1}}  ~ \| \tilde H \|_{L^{q_2}}
=
||F||_{L^{p,q_1}} ~ ||H||_{L^{p',q_2}},
\end{array}
\label{eqWindowedProduct5A}
\end{equation}
where
$\tilde F (\omega^{\prime}) = \| F (\cdot, \omega^{\prime}) \|_{L^p} $ and
$ \tilde H (\omega^{\prime}) = \| H (\cdot, \omega^{\prime}) \|_{L^{p'}} $.
By the assumptions it follows that
$\tilde F (\omega^{\prime}) \in L^{q_1} $, $
\tilde H (\omega^{\prime})  \in L^{q_2}, $ so the
last inequality in \eqref{eqWindowedProduct5A} follows from the  Young inequality.

\par

Therefore, $A_{F,H} \in \mathcal{F} L ^{q_0} (\mathbb{R}^d),$ and by \eqref{eq:fourierlebesguesembedding} it follows that
$ A_{F,H} \in  M^{ r, q_0}  (\mathbb{R}^d) $ when $r \geq \text{max} \{ q_0 , q' _0  \}.$ Thus we obtain
\begin{equation*}
\begin{array}{c}
||F \natural_{g} H||_{L^{r,q_0}}
= \| V_g  A_{F,H} \|_{L^{r,q_0}} = \|  A_{F,H} \|_{M^{r,q_0}} \precsim \\
\precsim  \|  A_{F,H} \|_{\mathcal{F} L ^{q_0}} = \| \hat A \|_{L ^{q_0}} \precsim
||F||_{L^{p,q_1}} ~ ||H||_{L^{p',q_2}}.
\end{array}
\end{equation*}

{\em 2.} By the homeomorphism  $  A_{F,H} \mapsto  \langle D \rangle ^{s}  A_{F,H} $ from $M^{r,q} _{0,s}$ to
$M^{r,q} _{0,0}$ we have
$$
\| A_{F,H} \|_{M^{r,q} _{0,s}} = \| \langle D \rangle^s A_{F,H} \|_{M^{r,q} _{0,0}} = \| \langle D \rangle^s A_{F,H} \|_{M^{r,q}},
$$
and the embedding $ \mathcal{F} L ^{q} (\mathbb{R}^d)
\hookrightarrow M^{ r, q}  (\mathbb{R}^d),$
with $ r \geq \text{max} \{ q, q'  \},$ gives
\begin{equation*}
\begin{array}{c}
\| A_{F,H} \|_{M^{r,q_0} _{0,s}} = \| \langle D \rangle^s A_{F,H} \|_{M^{r,q_0} }
\precsim \| \langle D \rangle^s A_{F,H} \|_{ \mathcal{F} L ^{q_0} } = \\
\\
= \| \widehat{ \langle D \rangle^s A_{F,H}} \|_{  L ^{q_0} }
=  \|  \langle \xi \rangle^s \hat A_{F,H} \|_{  L ^{q_0} } =
\|   \hat A_{F,H}\|_{  L ^{q_0} _s},
\end{array}
\end{equation*}
so that
$$
\| A_{F,H} \|_{M^{r,q_0} _{0,s}} \precsim
\|   \hat A_{F,H}\|_{L ^{q_0} _s}.
$$

Now we proceed as in {\em 1.},
\begin{equation}
\begin{array}{c}
||F \natural_{g} H ||_{L^{r,q_0} _{0,s}}
= \| A_{F,H} \|_{M^{r,q_0} _{0,s}} \precsim
\| \hat A_{F,H} \|_{L ^{q_0} _s} \leq \\
   \\
\leq
(\int_{\mathbb{R}^{d}} (\int_{\mathbb{R}^{2d}} | F(x^{\prime}, \omega^{\prime}) H(x^{\prime}, \xi- \omega^{\prime})|
\langle x' \rangle^{t_1 + t_2} dx^{\prime} d \omega^{\prime} )^{q_0} \langle \xi \rangle^s d\xi )^{1/q_0} \leq \\
\\
\leq
(\int_{\mathbb{R}^{d}}
(\int_{\mathbb{R}^{d}}
\| F (\cdot, \omega^{\prime}) \|_{L^p _{t_1}}  \| H (\cdot, \xi - \omega^{\prime}) \|_{L^{p'} _{t_2}} d \omega^{\prime} )^{q_0}
\langle \xi \rangle^s d\xi )^{1/q_0} = \\
\\
=
(\int_{\mathbb{R}^{d}} | (\tilde F * \tilde H) ( \xi)|^{q_0}
\langle \xi \rangle^s  d\xi )^{1/q_0}
\precsim \\
\\
\precsim \| \tilde F \|_{L^{q_1} _{s_1}}  ~ \| \tilde H \|_{L^{q_2}  _{s_2}}
=
\|F \|_{L^{p,q_1}  _{t_1, s_1}} ~ \|H \|_{L^{p',q_2}  _{t_2, s_2}},
\end{array}
\label{eqWindowedProduct6A}
\end{equation}
where
$\tilde F (\omega^{\prime}) = \| F (\cdot, \omega^{\prime}) \|_{L^p _{t_1}} $ and
$ \tilde H (\omega^{\prime}) = \| H (\cdot, \omega^{\prime}) \|_{L^{p'} _{t_2}} $, $ t_1 + t_2 \geq 0$,
and we used the Young inequality for weighted Lebesgue spaces
$$
L^{q_1} _{s_1} *  L^{q_2}  _{s_2} \subseteq L^{q_0}  _{s},
$$
when $ R(q) = 0,$  $ s_1+s_2 \geq 0,$ and $ \text{min} \{s_1, s _2 \} \geq s$,
see e.g. (0.1)--(0.3) in \cite{Toft1} (and also \cite[Theorem 2.2]{Toft1} for a more general situation).
\end{proof}

\begin{remark} \label{rem:extensions}
The results of Theorems \ref{TheoremModulationSpace4} and
 \ref{thm:extension-version1} are different, although partially overlapping. Let us briefly comment on unweighted case for simplicity.
If in  Theorem  \ref{thm:extension-version1}  we assume that
 $ F = V_{g_1} f $ and $ H =  V_{g_2} h $ for some $f,h \in  S_0 (\mathbb{R}^{d})$ and if
$g_1,g_2,g \in S_0 (\mathbb{R}^{d}) \backslash \left\{0 \right\}$, then by Theorem \ref{theoremWindowedProduct1}
it follows that
$$
F \natural_{g} H = V_{g_1} (f) \natural_{g} V_{\overline{g_2}} (h) =
\langle g_2,g_1  \rangle_{L^2 (\mathbb{R}^d)} V_{g} (f \cdot h),
$$
and Theorem \ref{thm:extension-version1} can be viewed as a special case of Theorem \ref{TheoremModulationSpace4}, when  $1\leq p,q_1,q_2 < \infty$.

\par

On the other hand, if the assumptions of Theorem  \ref{thm:extension-version1} are fulfilled
and if $g_1,g_2,g \in S_0 (\mathbb{R}^{d}) \backslash \left\{0 \right\}$, then
the extension of the Gabor product holds also when $ F \in L^{p,q_1} (\mathbb{R}^{2d})
\setminus \mathcal{V}_{g_1}^{p,q_1} (\mathbb{R}^{2d}) $ and
$ H \in L^{p',q_2}  (\mathbb{R}^{2d}) \setminus
 \mathcal{V}_{\overline{g_2}} ^{p',q_2} (\mathbb{R}^{2d}) $ which is not covered by Theorem \ref{TheoremModulationSpace4}. Moreover, unlike Theorem \ref{TheoremModulationSpace4}, Theorem  \ref{thm:extension-version1} holds even if (some of) the involved Lebesgue exponents are infinite.
\end{remark}

\begin{remark}\label{RemarkModulationSpace5}
In view of (\ref{eqRelationSTFTCrossWignerFunction}), both  Theorem \ref{TheoremModulationSpace4} and  Theorem
  \ref{thm:extension-version1}  hold when the Gabor product is replaced by the Wigner product.
\end{remark}

\subsection{Algebras} \label{subsec:algebras}

Under the conditions of Corollary \ref{corollaryHomomorphism}, the product (\ref{eqHomomorphism1}) can be used to define algebras of functions in $\mathbb{R}^{2d}$ from algebras in $\mathbb{R}^d$. Here are two examples.

\begin{example}\label{ExampleClosedAlgebra1}
The space $\mathcal S (\mathbb{R}^d)$ of test functions is an algebra under pointwise multiplication. Let $V_g
\left(\mathcal S (\mathbb{R}^d) \right)$ denote its range under
the action of the STFT $V_g$. Then $V_g \left(\mathcal S
(\mathbb{R}^d) \right)$ is an algebra with the product $ \natural_g $ given by
(\ref{eqHomomorphism1}). The same is true if $\mathcal S (\mathbb{R}^d)$ is replaced by $\mathcal S^{(1)} (\mathbb{R}^d)$.
\end{example}

\begin{example}\label{ExampleClosedAlgebra2}
Another interesting example is Feichtinger's algebra $S_0 (\mathbb{R}^d)$. Recall that by Lemma \ref{lm:Szero} we have
\begin{equation}
S_0 (\mathbb{R}^d) \subset C^0 (\mathbb{R}^d) \cap L^1 (\mathbb{R}^d) \cap \mathcal{F}\left( L^1 (\mathbb{R}^d)\right) \cap L^2(\mathbb{R}^d).
\label{eqFeichtingerAlgebra1}
\end{equation}
Let then $g \in S_0 (\mathbb{R}^d)$ with $\overline{g}=g$ and put $\mathcal{A}_0 (\mathbb{R}^{2d}) = V_g \left(S_0 (\mathbb{R}^d) \right) \subset L^1 (\mathbb{R}^{2d})$. In view of  Theorem \ref{theoremWindowedProduct1}
(where (\ref{eqFeichtingerAlgebra1}) is used) we conclude that $\mathcal{A}_0 (\mathbb{R}^{2d})$ is an algebra with respect to the product $\natural_g$.

On the other hand, if $g$ is in addition an even function, then $\mathcal{A}_0 (\mathbb{R}^{2d})$ is an algebra with respect to the product:
\begin{equation*}
\begin{array}{c}
\left(F \diamondsuit_g H \right) (x, \omega) := \int_{\mathbb{R}^{3d}} \overline{g(x^{\prime} + x^{\prime \prime} -x)} F(x^{\prime}, \omega^{\prime}) \times \\
 \\
 \times H (x^{\prime \prime} , \omega^{\prime}) e^{2 i \pi (\omega^{\prime} - \omega) \cdot (x^{\prime}+ x^{\prime \prime})} dx^{\prime} dx^{\prime \prime } d \omega^{\prime}.
\end{array}
\label{eqFeichtingerAlgebra2}
\end{equation*}
This product is defined for e.g $ F,H \in S_0 (\mathbb{R}^{2d})$ and it extends to appropriate mixed-norm spaces as in Subsection \ref{subsec:extension}.

Indeed, we have for $||g||_{L^2}=1$:
\begin{equation*}
\begin{array}{c}
V_g (f \star h) (x, \omega) = e^{-2 i \pi x \cdot \omega}V_{\widehat{g}} \left( (f \star h)^{\widehat{}} \right) (\omega, -x)=\\
\\
= e^{-2 i \pi x \cdot \omega}V_{\widehat{g}} \left( \widehat{f} \cdot \widehat{h} \right) (\omega, -x)=\\
\\
= e^{-2 i \pi x \cdot \omega} \left(V_{\widehat{g}} \widehat{f} \natural_{\widehat{g}}V_{\overline{\widehat{g}}} \widehat{h}\right) (\omega, -x) =\\
\\
= e^{-2 i \pi x \cdot \omega} \int_{\mathbb{R}^{3d}} \overline{g(x^{\prime} + x^{\prime \prime} -x)} V_{\widehat{g}} \widehat{f} (\omega^{\prime}, - x^{\prime}) \times \\
\\
\times V_{\overline{\widehat{g}}} \widehat{h} (\omega^{\prime}, - x^{\prime \prime}) e^{- 2 i \pi \omega \cdot (x^{\prime} + x^{\prime \prime} -x)} dx^{\prime} d x^{\prime \prime } d \omega^{\prime} =\\
\\
= \int_{\mathbb{R}^{3d}} \overline{g(x^{\prime} + x^{\prime \prime} -x)} V_g f (x^{\prime} , \omega^{\prime}) \times\\
\\
\times V_{\mathcal{I} \overline{g}} h (x^{\prime \prime} , \omega^{\prime}) e^{2 i \pi (\omega^{\prime}- \omega) \cdot (x^{\prime} + x^{\prime \prime})}  dx^{\prime} d x^{\prime \prime } d \omega^{\prime} =\\
\\
= \left( V_g f \diamondsuit_g  V_{\mathcal{I} \overline{g}} h  \right) (x, \omega).
\end{array}
\label{eqFeichtingerAlgebra3}
\end{equation*}

Consequently, $V_g (f \star h)=  V_g f \diamondsuit_g  V_g h$ if $g \in  S_0(\mathbb{R}^d)$ is real and even.
\end{example}

From this algebraic point of view, it is useful to determine an involution suitable for the products $\natural_g$ and $\sharp_g$ (see \cite[Section 1.1]{Folland1}).
First we note that the complex conjugate of  $F \in (\mathcal S^{(1)})^{\prime} (\mathbb{R}^{2d})$ is given by
$$
< \overline{F}, \Phi > := \overline{< F, \overline{\Phi}>}, \qquad \Phi \in \mathcal S^{(1)} (\mathbb{R}^{2d}),
$$
and the action of $\mathcal{Z}_{\xi}$ (see  (\ref{eqWindowedProduct6})) on $F\in (\mathcal S^{(1)})^{\prime} (\mathbb{R}^{2d})$ is given by
$$
< \mathcal{Z}_{\xi} F, \Phi > := < F, \mathcal{Z}_{\xi} \Phi>, \qquad \Phi \in \mathcal S^{(1)} (\mathbb{R}^{2d}).
$$

\begin{definition}\label{definitionInvolution}
Let $F \in (\mathcal S^{(1)})^{\prime} (\mathbb{R}^{2d})$. We define the operation
\begin{equation}
F(x, \omega)^{\ast} = \overline{F (x, - \omega)} = \overline{\left(\mathcal{Z}_0 F \right) (x, \omega)},
\label{eqInvolution}
\end{equation}
where $\mathcal{Z}_{\xi}$ is the operator given by (\ref{eqWindowedProduct6}).
\end{definition}

\begin{proposition}\label{PropositionInvolution}
Let $\mathcal{B} (\mathbb{R}^{2d})  $ be an algebra with respect to the product $\natural_g$ or $\sharp_g$ with $\overline{g}=g$ and $||g||_{L^2}=1$. The operation $\ast$ given by \eqref{eqInvolution} is an involution in $\mathcal{B} (\mathbb{R}^{2d})$.
\end{proposition}

\begin{proof}
Clearly,
\begin{equation*}
\left(F+G \right)^{\ast}=F^{\ast}+G^{\ast}, \hspace{1 cm} (\alpha F)^{\ast} = \overline{\alpha} F^{\ast}, \hspace{1 cm} (F^{\ast})^{\ast} =F
\label{eqProofInvolution1}
\end{equation*}
for all $F \in (\mathcal S^{(1)})^{\prime} (\mathbb{R}^{2d}) $, and $\alpha \in \mathbb{C}$.

Finally, let $F,H \in \mathcal{B} (\mathbb{R}^{2d})$. Then
\begin{equation*}
\begin{array}{c}
\left(F\natural_g H (x, \omega) \right)^{\ast} = \overline{F\natural_g H } (x, - \omega)=\\
\\
= \int_{\mathbb{R}^{3d}} \widehat{g} (\omega^{\prime} + \omega^{\prime \prime} + \omega) \overline{F(x^{\prime}, \omega^{\prime})} ~ \overline{H(x^{\prime}, \omega^{\prime \prime})} e^{-2 \pi i x \cdot(\omega^{\prime} + \omega^{\prime \prime} + \omega)} d x^{\prime} d \omega^{\prime} d \omega^{\prime \prime} =\\
\\
= \int_{\mathbb{R}^{3d}} \widehat{g} (- \omega^{\prime} - \omega^{\prime \prime} + \omega) \overline{H(x^{\prime}, - \omega^{\prime})} ~ \overline{F(x^{\prime}, - \omega^{\prime \prime})} e^{-2 \pi i x \cdot(- \omega^{\prime} - \omega^{\prime \prime} + \omega)} d x^{\prime} d \omega^{\prime} d \omega^{\prime \prime}
\end{array}
\label{eqProofInvolution2}
\end{equation*}
where we performed the substitutions $\omega^{\prime} \to - \omega^{\prime \prime}$ and $\omega^{\prime \prime} \to - \omega^{\prime}$ in the last step. Since, $\overline{g}=g$, we conclude that $\widehat{g} (- \omega) = \overline{\widehat{g} (\omega)}$ and so
\begin{equation*}
\left(F\natural_g H  \right)^{\ast} = H^{\ast} \natural_g F^{\ast}
\label{eqProofInvolution3}
\end{equation*}
and the result follows.

The proof for the Wigner product $\sharp_g$ is similar.
\end{proof}

Definition \ref{definitionInvolution} is obviously motivated by the fact that (cf.(\ref{eq2.5.1}))
\begin{equation*}
V_g \overline{f} (x, \omega) =  \overline{V_{ \overline{g}} f (x, - \omega)} =\overline{V_{g} f (x, - \omega)} = V_g f(x, \omega)^{\ast},
\label{eqInvolution2}
\end{equation*}
if $ \overline{g}=g$. Similarly from (\ref{eq2.4.3}) $W_g \overline{f} (x, \omega) =W_g f(x, \omega)^{\ast}$.

\section{Phase-space representations of the NLSE} \label{sec:3}

In this section we derive three different representations of the cubic NLSE in phase-space. The first two are obtained via  windowed transforms from Section \ref{sec:2}, and the third representation is given by using the Wigner transform. The resulting equation resembles  the Boltzmann equation.

Let $I$ be some open interval in $\mathbb{R}$. For some normed space $X$ with norm $|| \cdot||_X$, we denote by $C(I, X)$ the set of continuous functions from $I$ to $X$.

The Laplacian is given by $\Delta = \sum_{j=1}^d \frac{\partial^2}{\partial x_j^2}$. Finally, let $\psi : \mathbb{R}^d \times I \to \mathbb{C}$. The elliptic NLSE is given by
\begin{equation}
i \frac{\partial \psi}{\partial t} + \Delta \psi + \lambda | \psi|^{2 \sigma} \psi =0
\label{eqNLSE2}
\end{equation}
for an attracting $(\lambda =+1)$ or repulsive $(\lambda =-1)$ power-law nonlinearity. It is subject to the initial condition:
\begin{equation*}
\psi (x,0)= \varphi (x).
\label{eqNLSE3}
\end{equation*}
We then have the following theorem for the existence of solutions in Sobolev space $H^1 (\mathbb{R}^d) $ \cite{Sulem} (cf. Example \ref{ExampleModulationSpace1} for the definition of $H^1 (\mathbb{R}^d) $).

\begin{theorem}\label{theoremNLSE1}
For $0 \le \sigma < \frac{2}{d-2}$ (no conditions on $\sigma$ when $d=1$ or $d=2$) and an initial condition $\varphi \in H^1 (\mathbb{R}^d)$, there exists, locally in time, a unique maximal solution $\psi$ in $ C \left( (- T^*,T^*), H^1 (\mathbb{R}^d) \right)$, where maximal means that if $T^* < \infty$, then $|| \psi||_{H^1} \to \infty$ as $t$ approaches $T^*$. In addition, $\psi$ satisfies the probability and energy conservation laws:
\begin{equation*}
\begin{array}{l}
\mathcal{P} \left[\psi \right] := \int_{\mathbb{R}^d} |\psi (x,t)|^2 d x = \mathcal{P} \left[\varphi \right]\\
\\
H \left[\psi \right] := \int_{\mathbb{R}^d} \left(| \nabla \psi (x,t)|^2 - \frac{\lambda}{\sigma +1} |\psi (x,t)|^{2 \sigma +2} \right) d x = H \left[\varphi \right]
\end{array}
\label{eqNLSE4}
\end{equation*}
and depends continuously on the initial condition $\varphi \in H^1 (\mathbb{R}^d)$.

Moreover, if the initial condition $\varphi$ belongs to the space
\begin{equation*}
\Sigma = \left\{f \in H^1 (\mathbb{R}^d): ~ |x f(x)| \in L^2 (\mathbb{R}^d) \right\}
\label{eqNLSE5}
\end{equation*}
of functions in $H^1 (\mathbb{R}^d) $ with finite variance, the above maximal solution belongs to $C \left((-T^*, T^*), \Sigma \right)$. The variance
\begin{equation*}
V(t) := \int_{\mathbb{R}^d} |x|^2 | \psi (x,t)|^2 dx
\label{eqNLSE6}
\end{equation*}
belongs to $C^2 \left(-T^*, T^* \right)$ and satisfies the identity:
\begin{equation*}
\frac{d^2 V}{dt^2} = 8 H - 4\lambda \frac{d \sigma -2}{ \sigma+1} \int_{\mathbb{R}^d} | \psi (x,t)|^{2 \sigma +2} dx.
\label{eqNLSE7}
\end{equation*}
\end{theorem}
We will henceforth focus on the cubic equation $(\sigma=1)$.

\subsection{The STFT and the windowed Wigner representations}

Using the following intertwining relations
\begin{equation*}
V_g (x_j \psi) = - \frac{1}{2 \pi i} \frac{\partial}{\partial \omega_j} V_g (\psi), \hspace{1 cm} V_g \left(- i \frac{\partial \psi}{\partial x_j} \right) = \left(2 \pi \omega_j -i \frac{\partial}{\partial x_j} \right) V_g \psi,
\label{eqIntertwiners}
\end{equation*}
$ j=1, \cdots, d$, we obtain upon application of the STFT $V_g$ to (\ref{eqNLSE2}) with $\sigma=1$:
\begin{equation*}
i \frac{\partial F}{\partial t} - \sum_{j=1}^d \left(2 \pi \omega_j - i \frac{\partial}{\partial x_j} \right)^2 F + \lambda F^{\ast} \natural_g F \natural_g F =0,
\label{eqSTFTCubicNLSE}
\end{equation*}
where $F(x, \omega,t) = V_g (\psi) (x,\omega,t)$, $F^{\ast} $ is given by \eqref{eqInvolution}, $g$ is real and $\natural_g$ is as in (\ref{eqWindowedProduct1}). Here,
$$
V_g (\psi) (x,\omega,t) = \int_{\mathbb{R}^d} \psi(y,t) \overline{g(y-x)} e^{-2 \pi i \omega \cdot y} dy, \qquad x,\omega \in \mathbb{R}^d, \; t \in\mathbb{R}.
$$

Similarly, the intertwining relations for the Wigner wave-packet transform are given by the so-called Bopp operators:
\begin{equation*}
\begin{array}{l}
W_g (x_j \psi) = \frac{1}{2} \left(x_j + \frac{i}{\pi} \frac{\partial}{\partial \omega_j} \right) W_g \psi, \\
\\
W_g \left(\frac{1}{2 i \pi} \frac{\partial}{\partial x_j} \psi \right) = \frac{1}{2} \left(\omega_j - \frac{i}{\pi} \frac{\partial}{\partial x_j}  \right) W_g \psi, \qquad
j=1, \cdots, d.
\end{array}
\label{eqBoppOps1}
\end{equation*}

This then leads to the following phase-space representation of the cubic NLSE:
\begin{equation*}
i \frac{\partial F}{\partial t} - \sum_{j=1}^d \left( \pi \omega_j - i \frac{\partial}{\partial x_j} \right)^2 F + \lambda F^{\ast} \sharp_g F \sharp_g F =0,
\label{eqBoppOps2}
\end{equation*}
where this time $F=W_g \psi$, $g$ is real and $\sharp_g$ is given by (\ref{eqWignerProduct1}).

\subsection{The Wigner-Moyal representation}

Our final goal is to derive the Wigner-Moyal equation for the associated Wigner function $W \psi$. The kinetic part $\Delta \psi$ is known to lead to a diffusive term of the form $4 \pi \omega \cdot \nabla_x W \psi (x, \omega)$.

Next, we consider the interaction term:
\begin{equation*}
\begin{array}{c}
i\lambda \int_{\mathbb{R}^d} e^{- 2 i \pi \omega \cdot y} \left\{ \psi \left(x + \frac{y}{2} \right) \left| \psi \left(x - \frac{y}{2} \right) \right|^2 \overline{\psi \left(x - \frac{y}{2} \right) } \right.\\
  \\
  \left. -\overline{\psi \left(x - \frac{y}{2} \right) }  \left| \psi \left(x + \frac{y}{2} \right) \right|^2  \psi \left(x + \frac{y}{2} \right)\right\}dy =\\
  \\
 =i \lambda \int_{\mathbb{R}^d} e^{- 2 i \pi \omega \cdot y} \psi \left(x + \frac{y}{2} \right)  \overline{\psi \left(x - \frac{y}{2} \right) } \\
 \\
 \int_{\mathbb{R}^d} \left\{ W \psi \left(x - \frac{y}{2} , \omega^{\prime \prime} \right) - W \psi \left(x + \frac{y}{2} , \omega^{\prime \prime} \right) d \omega^{\prime \prime} \right\} dy
\end{array}
\label{eqNLSE8}
\end{equation*}
where we used (\ref{eqWignerfunction3}).

Gathering all the results and using (\ref{eqWignerfunction6}), we finally obtain:
\begin{equation*}
\frac{\partial W \psi}{\partial t} = - 4 \pi \omega \cdot \nabla_x W \psi + \mathcal{Q} \left(W \psi, W \psi \right),
\end{equation*}
where the "collision" term is given by:
\begin{equation*}
\begin{array}{c}
\mathcal{Q} (F,G) (x, \omega) := 2^d i \lambda \int_{\mathbb{R}^{3 d}} e^{-4 i \pi y \cdot (\omega -\omega^{\prime})} F (x, \omega^{\prime}) \times \\
\\
 \times \left\{G \left(x + y , \omega^{\prime \prime} \right) - G \left(x - y , \omega^{\prime \prime} \right)  \right\}   d y d \omega^{\prime} d \omega^{\prime \prime}=\\
 \\
 = 2^{d+1}  \lambda \int_{\mathbb{R}^{3 d}}   \sin \left[ 4 \pi (\omega- \omega^{\prime}) \cdot (y-x) \right] F(x, \omega^{\prime}) G(y, \omega^{\prime \prime}) d y d \omega^{\prime} d \omega^{\prime \prime}
\end{array}
\label{eqNLSE9A}
\end{equation*}

We have thus far been somewhat informal in the derivation of the collision term. We may assume for the time being that $F,G \in S_0 (\mathbb{R}^{2d})$. To extend this definition, we shall assume that
\begin{equation}
F(x, \cdot) =F_x (\cdot) \in L^p (\mathbb{R}^d),
\label{eqCollisionTerm1}
\end{equation}
for almost all $x \in \mathbb{R}^d $ and $p \ge 1$, and
\begin{equation}
g_1 (x) = \int_{\mathbb{R}^d} G(x, \omega) d \omega \in \mathcal{F} L^q (\mathbb{R}^d), \hspace{1 cm} q \ge 1.
\label{eqCollisionTerm2}
\end{equation}
By a straightforward computation, we can then show that
\begin{equation*}
\mathcal{Q} (F,G) (x,\omega) = 2^{d/2} i \lambda \left[\left(F_x \star D_{\frac{1}{2}} M_x \widehat{g}_1  \right) (\omega) - \left(F_x \star \mathcal{I} D_{\frac{1}{2}} M_x \widehat{g}_1  \right) (\omega)  \right],
\label{eqCollisionTerm3}
\end{equation*}
for almost all $x \in \mathbb{R}^d$.

Consequently, by Young's Theorem (Theorem \ref{TheoremYoung}), we conclude that under the conditions (\ref{eqCollisionTerm1},\ref{eqCollisionTerm2}), we have:
\begin{equation*}
\mathcal{Q} (F,G) (x,\cdot) \in L^r (\mathbb{R}^d),
\label{eqCollisionTerm4}
\end{equation*}
for almost all $x \in \mathbb{R}^d$ and $r$ such that $\frac{1}{p} + \frac{1}{q} = 1+ \frac{1}{r}$.

The derived equation leads to several conservation laws. We illustrate here the conservation of the normalization. Quantum mechanically this corresponds to the conservation of probability. We remark that if $\psi \in H^1 (\mathbb{R}^d)$, then $W \psi \in L_{0,2}^{1,1} (\mathbb{R}^{2d})$.

We start with the following Lemma.

\begin{lemma}\label{Lemma2}
Suppose that $F,G \in L_{0,2}^{1,1} (\mathbb{R}^{2d})$ and (\ref{eqCollisionTerm2}) holds for $q=1$. Then the following identity holds:
\begin{equation}
\int_{\mathbb{R}^d} \mathcal{Q} (F,G) (x, \omega) d \omega =0,
\label{eqCollisionTerm4.1}
\end{equation}
for a.e. $x \in \mathbb{R}^d$.
\end{lemma}

\begin{proof}
Since $F \in L_{0,2}^{1,1} (\mathbb{R}^{2d}) \subset L^1 (\mathbb{R}^{2d})$, we conclude that $F_x (\cdot) \in L^1 (\mathbb{R}^{d})$ for a.e. $x \in \mathbb{R}^d$. From the previous analysis it follows that $\mathcal{Q} (F,G) (x,\cdot) \in L^1 (\mathbb{R}^{d})$ for a.e. $x \in \mathbb{R}^d$. By Fubini's Theorem, we have:
\begin{equation*}
\begin{array}{c}
\int_{\mathbb{R}^d} \mathcal{Q} (F,G) (x, \omega) d \omega =0,
\end{array}
\label{eqCollisionTerm5}
\end{equation*}
for a.e. $x \in \mathbb{R}^d$, which proves (i).
\end{proof}

Let us now comment on the conservation laws for the equation
\begin{equation}
\frac{\partial F}{\partial t} + 4 \pi \omega \cdot \nabla_x F = \mathcal{Q} (F,F).
\label{eqCollisionTerm6}
\end{equation}
We emphasize again the strong formal similarity between this equation and Boltzmann's equation.

Under the conditions of Lemma \ref{Lemma2} for $F$, we have:
\begin{equation*}
\begin{array}{c}
\frac{\partial}{\partial t} \int_{\mathbb{R}^{2d}} F (x, \omega,t) d x d \omega=\\
\\
= \int_{\mathbb{R}^{2d}} \left(- 4 \pi \omega \cdot \nabla_x F (x, \omega,t) + \mathcal{Q} (F,F) (x, \omega,t \right) d x d \omega.
\end{array}
\label{eqNLSE12}
\end{equation*}
The first term vanishes as it is the integral of a total derivative. The second term also vanishes in view of identity (\ref{eqCollisionTerm4.1}).

It may be shown that under suitable regularity conditions, we have conservation of energy:
\begin{equation*}
\frac{\partial}{\partial t} \int_{\mathbb{R}^{2d}} \left( \omega^2 - \frac{\lambda}{8 \pi^2} \int_{\mathbb{R}^d} F (x, \omega^{\prime},t) d \omega^{\prime} \right) F (x, \omega,t) dx d \omega =0.
\label{eqNLSE13}
\end{equation*}

We conclude with a brief comment on the solutions of this equation. Suppose we choose as the initial distribution the Wigner function $F(x,\omega,0)= W \varphi (x, \omega)$ associated with some wave function $\varphi \in H^1 (\mathbb{R}^d)$ as in Theorem \ref{theoremNLSE1}. Then the solution of (\ref{eqCollisionTerm6}) is given by $W \psi (x,\omega,t)$, where $\psi(x,t)$ is the solution of the cubic NLSE (for $t \in I$). In a future work, we shall study the existence of solutions of (\ref{eqCollisionTerm6}) for initial distributions which are not Wigner functions.

\section*{Acknowledgements}

The authors are grateful to H. G. Feichtinger for valuable comments, suggestions and discussions which helped us in improving the first version of the manuscript.

The work of N. Teofanov is partially supported  by projects TIFREFUS Project DS 15, and MPNTR of Serbia Grant No. Grant No. 451--03--68/2022--14/200125.

\vspace{2cm}

**********************************************************************************************************************************************************************************************************

\textbf{Author's addresses:}

\begin{itemize}
\item \textbf{Nuno Costa Dias and Jo\~ao Nuno Prata: }Grupo de F\'{\i}sica
Matem\'{a}tica, Departamento de Matem\'atica, Faculdade de Ci\^encias, Universidade de Lisboa, Campo Grande, Edif\'{\i}cio C6, 1749-016 Lisboa, Portugal and Escola Superior N\'autica Infante D. Henrique. Av.
Eng. Bonneville Franco, 2770-058 Pa\c{c}o d'Arcos, Portugal.

\item \textbf{Nenad Teofanov:} Department of Mathematics and Informatics, Faculty of Sciences,
University of Novi Sad, Trg D. Obradovi\'ca 4, 21000 Novi Sad, Serbia.
\end{itemize}

**********************************************************************************************************************************************************************************************************

\end{document}